\documentclass[
  10pt,
  nofootinbib,
  superscriptaddress,
  onecolumn,
  amsfonts,
  amssymb,
  amsmath,
  aps,
  pra,
  floatfix
]{revtex4-2}
\usepackage{mathtools}
\usepackage{mathrsfs}
\usepackage{amsthm}
\usepackage{bm}
\usepackage{array}
\usepackage{dcolumn}
\usepackage{multirow}
\usepackage{booktabs}
\usepackage{enumitem}
\usepackage{graphicx}
\usepackage{tikz}
\usetikzlibrary{quantikz2}
\usepackage{algpseudocode}
\usepackage{csquotes}
\usepackage{comment}
\usepackage{xcolor}
\usepackage[normalem]{ulem}
\usepackage{orcidlink}
\usepackage{xurl}
\usepackage{hyperref}
\hypersetup{
  colorlinks=true,
  citecolor=cyan,
  urlcolor=cyan,
  linkcolor=cyan,
  breaklinks=true,
}
\newtheoremstyle{break}
  {} {} {\itshape} {} {\bfseries} {.} {\newline} {\thmname{#1}\thmnumber{ #2}\thmnote{ (\bfseries #3)}}
\theoremstyle{break}
\newtheorem{theorem}{Theorem}
\newtheorem{proposition}{Proposition}

\newtheorem{assumption}{Assumption}

\DeclareMathOperator{\Tr}{Tr}

\DeclareMathOperator{\Cov}{Cov}

\DeclareMathOperator{\diag}{diag}
\DeclareMathOperator{\Bdiag}{Bdiag}
\DeclareMathOperator*{\argmax}{arg\,max}

\let\oldd\d \renewcommand{\d}{\ifmmode\mathrm{d}\else\oldd\fi}
\let\oldi\i \renewcommand{\i}{\ifmmode\mathrm{i}\else\oldi\fi}
 
\let\Re\relax \DeclareMathOperator{\Re}{Re}
\let\Im\relax \DeclareMathOperator{\Im}{Im}
\let\ket\relax
\let\bra\relax

\let\ketbra\relax
\newcommand{\ket}[1]{| {#1} \rangle}
\newcommand{\bra}[1]{\langle {#1} |}

\newcommand{\ketbra}[2]{| {#1} \rangle\langle {#2} |}

\newcommand{\ip}[2]{\langle {#1} | {#2} \rangle}
\newcommand{\norm}[1]{\left\| {#1} \right\|}
\newcommand{\bftheta}{\bm{\theta}}
\newcommand{\Ex}{\mathbb{E}}

\newcommand{\veps}{\varepsilon}

\renewcommand{\i}{\mathrm{i}}
\newcommand{\dd}{\mathrm{d}}

\newcommand{\bbR}{\mathbb{R}}

\makeatletter
\@tfor\Letter:=ABCDEFGHIJKLMNOPQRSTUVWXYZ\do{%
  \expandafter\edef\csname rm\Letter\endcsname{\noexpand\mathrm{\Letter}}
  \expandafter\edef\csname bf\Letter\endcsname{\noexpand\mathbf{\Letter}}
  \expandafter\edef\csname sf\Letter\endcsname{\noexpand\mathsf{\Letter}}
  \expandafter\edef\csname cal\Letter\endcsname{\noexpand\mathcal{\Letter}}
  \expandafter\edef\csname scr\Letter\endcsname{\noexpand\mathscr{\Letter}}
}
\makeatother

\definecolor{lightblue}{rgb}{0.678, 0.847, 0.902}
\definecolor{lightgreen}{rgb}{0.565, 0.933, 0.565}
\definecolor{lightpink}{rgb}{1.000, 0.714, 0.757}

\begin{document}
\title{Quantum Federated Learning Based on Bures--Uhlmann Geometry\\ for Heterogeneous Noisy Clients}
\author{Haruki Emori\,\orcidlink{0009-0007-2264-9192}}
\email{emori.haruki.i8@elms.hokudai.ac.jp}
\affiliation{Graduate School of Information Science and Technology, Hokkaido University,
Kita 14, Nishi 9, Kita-ku, Sapporo, Hokkaido 060-0814, Japan}
\affiliation{RIKEN Center for Interdisciplinary Theoretical and Mathematical Sciences (iTHEMS),
RIKEN, 2-1 Hirosawa, Wako, Saitama 351-0198, Japan}
\author{Masaki Uchihara\,\orcidlink{0000-0002-3275-6625}}
\affiliation{Department of Diabetes, Endocrinology, and Metabolism,
National Center for Global Health and Medicine, Japan Institute for Health Security,
1-21-1 Toyama, Shinjuku-ku, Tokyo 162-8655, Japan}
\author{Yuuki Tokunaga\,\orcidlink{0000-0002-5167-4609}}
\affiliation{Graduate School of Information Science and Technology, Hokkaido University,
Kita 14, Nishi 9, Kita-ku, Sapporo, Hokkaido 060-0814, Japan}
\date{\today}

\begin{abstract}
Quantum federated learning enables collaborative model training across quantum devices without sharing raw data, and it faces the data and hardware heterogeneity inherent to noisy quantum devices.
Utilizing the quantum geometric tensor is a natural remedy, yet pure-state approaches and diagonal approximations discard the correlations that encode parameter incompatibility.
To address this, we extend the parameter-space geometry to the mixed states that noisy clients actually prepare.
The real part of the resulting mixed-state geometric tensor is the Bures metric, which measures how fast the physical state changes under parameter variation, and the imaginary part is the mean Uhlmann curvature, which quantifies the incompatibility of estimating multiple parameters simultaneously.
Accordingly, we employ the Bures metric as a local preconditioner and use the mean Uhlmann curvature to develop an achievable-precision aggregation rule that dynamically down-weights unreliable clients.
Furthermore, we establish theoretical guarantees by proving a convergence theorem and a variance-dominance proposition.
Empirical evaluations on a trapped-ion quantum emulator demonstrate that the proposed method maintains high accuracy across diverse device-heterogeneity conditions and outperforms standard federated averaging, whose accuracy degrades under strong noise.
\end{abstract}

\maketitle

\section{Introduction}
Machine learning increasingly relies on data that are private, regulated, or too large to centralize, such as clinical records held by separate hospitals, sensor logs on edge devices, or financial transactions inside different institutions.
Federated learning (FL) was introduced to train a shared model on such siloed data without moving the raw data off the owning device, in which each client optimizes a local model and shares only parameter updates with a coordinating server~\cite{McMahan17}.
This addresses societal needs for data privacy and regulatory compliance, and it inherits two structural difficulties.
First, client datasets are non-independent and non-identically distributed (non-IID), which slows or biases the global model.
Second, communication of model updates is costly and is itself a privacy attack surface.

Quantum federated learning (QFL) augments this picture with quantum resources.
Clients hold variational quantum circuits (VQCs)~\cite{Benedetti19,Cerezo21,Schuld21}, encode their classical data into quantum states, and train by measuring expectation values, and the server aggregates the circuit parameters~\cite{Chen21,Nguyen25,Sai25}.
The motivation is twofold.
On the modeling side, VQCs access high-dimensional Hilbert spaces and may represent correlations that are hard classically.
On the systems side, quantum communication primitives such as quantum key distribution can protect the shared updates with physical-layer security~\cite{Nguyen25}.
A comprehensive account of QFL architectures, communication schemes, optimization strategies, security mechanisms, and applications is given in the recent survey of Nguyen \emph{et al.}~\cite{Nguyen25}.
Prior QFL work has, for example, introduced Fisher-information--based aggregation for non-IID data (QFedFisher)~\cite{Bhatia25}, adaptive two-gain aggregation over noisy quantum channels~\cite{Nanayakkara25}, layerwise training for hardware-heterogeneous clients~\cite{Han25}, and the use of intrinsic hardware noise for differential privacy~\cite{Pokharel25}.

A difficulty specific to QFL, and largely absent from classical FL, is
\emph{device heterogeneity}.
In the noisy intermediate-scale quantum (NISQ) era each client runs on hardware with a different gate-error rate, coherence time, and number of measurement shots.
Two clients with identical data and different devices produce updates of different reliability, and naive averaging is influenced most by the worst devices.
Our numerical experiments make this concrete.
When the unreliable clients have depolarizing rates of $0.4$--$0.6$, federated averaging of the same ansatz reaches a round-averaged accuracy of $0.77$--$0.78$ on a ten-class digit task with a worst evaluated round as low as $0.67$, whereas the geometry-aware rule proposed here reaches $0.81$--$0.82$ with a worst round above $0.73$.

A geometry-aware treatment of the parameter space is a natural remedy.
The relevant geometric object is the quantum geometric tensor (QGT)
$Q_{ij}=g_{ij}-\tfrac{\i}{2}\omega_{ij}$~\cite{Provost80,Cheng10,Gao25,Liu25}, whose real part $g_{ij}$ is the Fubini--Study metric and whose imaginary part $\omega_{ij}$ is the Berry curvature.
A subtlety here motivates the present work.
For a \emph{pure} state the real part of the QGT equals one quarter of the quantum Fisher information matrix (QFIM), $F^{\mathrm Q}_{ij}=4g_{ij}$~\cite{Stokes20,Liu20}.
Any optimizer that preconditions the gradient with the real part of the \emph{pure-state} QGT is, by definition, the quantum natural gradient (QNG) of Stokes \emph{et al.}~\cite{Stokes20}, and transplanting it into a federated loop adds no geometric novelty on the local-optimization side.
The diagonal approximation commonly used to reduce the measurement cost of the metric discards the off-diagonal correlations that, as we show below, carry the quantum \emph{incompatibility} information once the state is mixed.
These two points limit any pure-state QGT approach to QFL.

We resolve both points by moving the geometry to the regime in which QFL clients \emph{actually} operate.
Because of hardware noise, a client prepares a mixed state $\rho_k(\bftheta)$ rather than a pure state $\ket{\psi(\bftheta)}$.
The geometry for a density matrix is the \emph{mixed-state} QGT, whose real part is the Bures metric, the minimal monotone metric on density matrices~\cite{Bures69,Uhlmann91,Petz96,Carollo20}, and whose imaginary part is the mean Uhlmann curvature~\cite{Uhlmann91,Carollo20}.
The Bures metric reduces to the Fubini--Study metric for pure states and to the classical Fisher matrix for commuting families, and it differs from both in the intermediate, $O(1)$-noise regime where QFL operates.
This is the open direction that Stokes \emph{et al.} flagged in their discussion as the candidate for noisy devices~\cite{Stokes20}, and we develop it in the QFL context.

\paragraph{Contributions.} (i) We identify the pure-state equivalence
QGT(real)$=$QFIM$=$QNG, and we extend QFL beyond it to the
mixed-state Bures--Uhlmann geometry.
(ii) We make the roles of the two parts precise.
The symmetric (Bures) part is a reparametrization-invariant local preconditioner, and the antisymmetric (Uhlmann) part enters, through the quantum Cram\'er--Rao bound~\cite{Helstrom69,Helstrom76,Holevo82,Paris09}, an \emph{achievable-precision} aggregation rule.
(iii) We replace the diagonal approximation by an exact block-diagonal
estimator with $O(L)$ overhead and bound its truncation error.
(iv) We derive the closed-form depolarizing attenuation $\kappa(p)$ and show that for non-global noise the correction is matrix-valued, and the full mixed-state geometry is needed.
(v) We prove a convergence theorem, and we prove in addition (Proposition~\ref{prop:dominance}) that the proposed aggregation never has a larger residual variance than federated averaging, with equality if and only if the devices are homogeneous.
The reduction of the aggregated variance is therefore a consequence of the analysis rather than a tuning artefact.
(vi) We evaluate on the full MNIST database (ten-class task, digits $0$--$9$), using the trapped-ion emulator \texttt{reimei-E} of the \texttt{reimei} quantum computer~\cite{Quantinuum} through Quantinuum Nexus under four device-heterogeneity conditions, averaging over three independently seeded federations.
The proposed method attains the highest round-averaged accuracy in every condition, and the size of its advantage over federated averaging is largest where the closed-form weight ratio predicts the strongest down-weighting and narrows toward parity where the devices are reliable.
The approach points toward robust, privacy-preserving collaborative learning on the heterogeneous quantum hardware of the emerging quantum internet.

\section{Preliminaries}

\subsection{Quantum federated learning}
\label{sec:qfl}
We first describe QFL quantitatively for readers familiar with quantum mechanics and not with FL.
A qubit state is a unit vector $\ket{\psi}=\alpha\ket{0}+\beta\ket{1}\in\mathbb{C}^2$ with $|\alpha|^2+|\beta|^2=1$, and an $n$-qubit register lives in $\mathcal{H}=(\mathbb{C}^2)^{\otimes n}$ of dimension $N=2^n$.
A variational quantum classifier maps a classical input $\mathbf{x}\in\bbR^{d_{\rm in}}$ to a parametrized state
\begin{equation}
\ket{\psi(\mathbf{x};\bftheta)}=U(\bftheta)\,U_{\rm Enc}(\mathbf{x})\ket{0}^{\otimes n},
\label{eq:vqc-state}
\end{equation}
where $U_{\rm Enc}(\mathbf{x})$ encodes the data by angle encoding, $R_y(x_j)=e^{-\i x_j Y/2}$, and $U(\bftheta)$ is a trainable unitary with real parameters $\bftheta\in\bbR^{d}$.
The model output is built from the expectation values of a fixed set of observables $\{P_f\}$,
\begin{equation}
\phi_f(\mathbf{x};\bftheta)=\Tr\!\big[P_f\,\ketbra{\psi(\mathbf{x};\bftheta)}{\psi(\mathbf{x};\bftheta)}\big],
\label{eq:vqc-output}
\end{equation}
read out by repeated measurement, with the set specified in Sec.~\ref{sec:readout}.
Gradients $\partial_i\phi_f$ are obtained on hardware by the parameter-shift rule~\cite{Mitarai18,Schuld19}.

In FL there are $K$ clients.
Client $k$ owns a private dataset $\mathcal{D}_k$ of size $n_k$ that \emph{never} leaves the device, and $\mathcal{D}=\bigsqcup_{k=1}^{K}\mathcal{D}_k$ with $|\mathcal{D}|=\sum_k n_k$.
With a per-sample loss $\ell$, here the softmax cross-entropy, the local objective is
\begin{equation}
f_k(\bftheta)=\frac{1}{n_k}\sum_{(\mathbf{x},y)\in\mathcal{D}_k}
\ell\big(y,\hat y(\mathbf{x};\bftheta)\big),
\label{eq:local-loss}
\end{equation}
and the federated goal is to minimize the global objective
\begin{equation}
F(\bftheta)=\sum_{k=1}^{K}p_k\,f_k(\bftheta),\qquad p_k=\frac{n_k}{|\mathcal{D}|}.
\label{eq:global-obj}
\end{equation}

Training proceeds in synchronous \emph{rounds}.
The server broadcasts the global parameters $\bftheta^{(t)}$, each client performs local optimization on $f_k$ and returns its updated parameters $\bftheta_k^{(t+1)}$, and the server aggregates them into $\bftheta^{(t+1)}$.
The canonical rule, federated averaging (QFedAvg), uses the data-size weights, $\bftheta^{(t+1)}=\sum_k p_k\,\bftheta_k^{(t+1)}$.
QFedFisher~\cite{Bhatia25} replaces $p_k$ by weights built from a diagonal classical Fisher information.
The central question of this paper is what aggregation weights are geometrically and statistically appropriate when clients are noisy and heterogeneous, and how the parameter-space geometry should shape the local update.

\subsection{Pure-state quantum geometric tensor}
\label{sec:pure}
We recall the QGT in a form that makes its meaning transparent, following the
geometric formulation of quantum mechanics~\cite{Cheng10,Provost80}.
Consider a smooth family of pure states $\ket{\psi(\bftheta)}$.
The naive distance obtained from $\ip{\partial_i\psi}{\partial_j\psi}$ is not gauge invariant, because a local phase change $\ket{\psi}\to e^{\i\alpha(\bftheta)}\ket{\psi}$, which leaves all physical observables unchanged, alters it.
Removing the phase (Berry) connection $A_i=\i\ip{\psi}{\partial_i\psi}\in\bbR$ yields the gauge-invariant quantum geometric tensor on the projective space of rays~\cite{Provost80},
\begin{equation}
Q^{\rm Pure}_{ij}(\bftheta)=\ip{\partial_i\psi}{\partial_j\psi}
-\ip{\partial_i\psi}{\psi}\ip{\psi}{\partial_j\psi}
=g_{ij}-\tfrac{\i}{2}\,\omega_{ij}.
\label{eq:qgt-pure}
\end{equation}
Because the Hilbert-space inner product is Hermitian, the real part $g_{ij}=\Re Q^{\rm Pure}_{ij}$ is symmetric and the imaginary part $\omega_{ij}=-2\,\Im Q^{\rm Pure}_{ij}$ is antisymmetric.
The real part is the
Fubini--Study metric~\cite{Provost80}, and it measures the rate at which the
physical state changes under parameter variation, because the fidelity obeys~\cite{Provost80}
\begin{equation}
|\ip{\psi(\bftheta)}{\psi(\bftheta+\dd\bftheta)}|^2
= 1-\sum_{ij}g_{ij}\,\dd\theta_i\,\dd\theta_j+O(\dd\theta^3).
\label{eq:fidelity-expand}
\end{equation}
The antisymmetric part $\omega_{ij}$ is the Berry curvature, the flux of the emergent gauge field $A_i$ associated with the geometric (Berry) phase~\cite{Berry84}.
Thus $g$ describes \emph{how far} the state moves, while $\omega$ describes the \emph{phase holonomy} accumulated, and one needs $g$ to endow the parameter manifold with a notion of length and steepest descent.

For learning, the relevance of $g$ is twofold.
It is the unique unitarily invariant metric on pure states, and it provides a parametrization-independent notion of the natural descent direction, and it equals one quarter of the pure-state quantum Fisher information matrix, $F^{\mathrm Q}_{ij}=4g_{ij}$~\cite{Liu20}, which links geometry to estimation precision.
The quantum natural gradient (QNG)~\cite{Stokes20,Amari98} uses this by preconditioning the Euclidean gradient with $g$,
\begin{equation}
\bftheta^{(t+1)}=\bftheta^{(t)}-\eta\,\big(g(\bftheta^{(t)})+\veps\openone\big)^{-1}\nabla\mathcal{L}(\bftheta^{(t)}),
\label{eq:qng}
\end{equation}
where $\mathcal{L}(\bftheta)$ is the differentiable training objective minimized locally, in the federated setting the client empirical risk $f_k(\bftheta)$ of Eq.~\eqref{eq:local-loss}, and $\nabla\mathcal{L}(\bftheta^{(t)})$ is its Euclidean gradient at the current iterate $\bftheta^{(t)}$.
Here $\eta>0$ is the learning rate and $\veps>0$ is a Tikhonov regularizer that guarantees the inverse $(g+\veps\openone)^{-1}$ exists even when $g$ is rank deficient.
Equation~\eqref{eq:qng} corrects for the curvature of the state manifold and mitigates barren plateaus~\cite{McClean18,Wang21}.

The observation for QFL is that Eq.~\eqref{eq:qng} \emph{is} QNG, and any method that preconditions with the real part of the pure-state QGT reproduces it exactly.
A pure-state QGT-based QFL therefore cannot claim novelty at the level of the local update unless the geometry itself is changed.
This is what the mixed-state extension below accomplishes.

\subsection{Bures metric and mean Uhlmann curvature for mixed states}
\label{sec:mixed}
On real hardware client $k$ realizes a mixed state produced by a noise channel $\mathcal{N}_k$ rather than the pure state \eqref{eq:vqc-state},
\begin{equation}
\rho_k(\bftheta)=\mathcal{N}_k\!\big(U(\bftheta)\ketbra{0}{0}U^\dagger(\bftheta)\big),
\qquad \Tr\rho_k^2<1.
\label{eq:rho}
\end{equation}
The geometric object appropriate to a density matrix is built from the symmetric
logarithmic derivative (SLD) $L_i$, defined by $\partial_i\rho_k=\tfrac12(L_i\rho_k+\rho_k L_i)$~\cite{Helstrom69,Helstrom76,Holevo82,Paris09}.
In the eigenbasis $\rho_k=\sum_n\lambda_n\ketbra{n}{n}$ the mixed-state QGT is
\begin{equation}
Q^{\rm Mix}_{k,ij}\coloneqq\tfrac14\Tr\!\big[\rho_k\,L_iL_j\big]
=g^{\rm Bur}_{k,ij}-\tfrac{\i}{2}\,\omega^{\rm Uhl}_{k,ij},
\label{eq:qgt-mix}
\end{equation}
whose real (Bures) and imaginary (mean Uhlmann curvature) parts have the explicit
forms~\cite{Carollo20}
\begin{align}
g^{\rm Bur}_{k,ij}&=\tfrac14 F^{\mathrm Q,\mathrm{SLD}}_{k,ij}
=\frac12\!\!\sum_{\lambda_m+\lambda_n>0}\!\!
\frac{\Re\!\big(\bra{m}\partial_i\rho_k\ket{n}\bra{n}\partial_j\rho_k\ket{m}\big)}
{\lambda_m+\lambda_n},
\label{eq:bures}\\
\omega^{\rm Uhl}_{k,ij}&=\tfrac{\i}{2}\Tr\!\big(\rho_k[L_i,L_j]\big)
=\!\!\sum_{\lambda_m+\lambda_n>0}\!\!
\frac{(\lambda_m-\lambda_n)\,\Im\!\big(\bra{m}\partial_i\rho_k\ket{n}\bra{n}\partial_j\rho_k\ket{m}\big)}
{(\lambda_m+\lambda_n)^2}.
\label{eq:uhl}
\end{align}
Equations~\eqref{eq:bures}--\eqref{eq:uhl} are the mixed-state generalizations of the real and imaginary parts of \eqref{eq:qgt-pure}.
The difference from the pure case is conceptual and quantitative, and it is summarized by three points.

\emph{(a) Distinctness from QFIM and classical Fisher.} By Petz's theorem~\cite{Petz96}, monotone metrics on density matrices are not unique, and the Bures metric is the \emph{minimal} one.
It reduces to the Fubini--Study metric in the pure limit $\rho_k\to\ketbra{\psi}{\psi}$ and to one quarter of the classical Fisher matrix in the commuting limit $[\rho_k,\partial_i\rho_k]=0$, and for an $O(1)$-noise state it coincides with \emph{neither}~\cite{Carollo20}.
QFL clients operate in this intermediate regime, and a natural choice of metric is $g^{\rm Bur}_k$ rather than the pure-state QGT used by QNG or the classical Fisher matrix used by QFedFisher~\cite{Bhatia25}.

\emph{(b) Quantum content in the imaginary part.} The mean Uhlmann
curvature \eqref{eq:uhl} vanishes for pure states and for commuting families and is nonzero whenever $[\rho_k,\partial_i\rho_k]\neq0$.
It measures the \emph{incompatibility} of estimating different parameters simultaneously and has no classical analogue~\cite{Carollo20}.
Discarding the off-diagonal elements by a diagonal approximation destroys this quantum information, which we avoid in Sec.~\ref{sec:method}.

\emph{(c) Separate roles of the two parts.} For a \emph{real} parameter update the descent direction solves $g\,\delta\bftheta=-\eta\nabla\mathcal{L}$, and any antisymmetric matrix $\omega$ contributes nothing because $\delta\bftheta^{\!\top}\omega\,\delta\bftheta=0$.
The imaginary part cannot, and must not, be inserted into the gradient step, which is why QNG uses \emph{only} the real part.
The imaginary part has a distinct statistical role.
When the SLDs do not commute, the attainable covariance of a multiparameter estimate exceeds the SLD quantum Cram\'er--Rao bound (the Holevo bound), and the excess is governed by
$\omega^{\rm Uhl}_k$~\cite{Carollo20,Paris09,Imai26PRL,Imai26ARX}.
We therefore use the Bures part as the local preconditioner and the Uhlmann part to quantify the \emph{reliability} of each client's reported update (Sec.~\ref{sec:agg}).
This division of labour uses both parts of the QGT and does not insert the imaginary part into the gradient step.

\section{Proposed method QFedQGT}
\label{sec:method}
The proposed scheme has two ingredients, a mixed-state natural-gradient local update that uses the Bures metric, and a server aggregation rule whose weights are the \emph{achievable precision} of each client.
Both are estimated on hardware by an exact block-diagonal procedure, and the diagonal approximation is never used.

We first fix notation.
The noiseless (pure) QGT real part of the shared ansatz $\ket{\psi(\bftheta)}=U(\bftheta)\ket{0}^{\otimes n}$ is
\begin{equation}
\big[G^{\rm Pure}(\bftheta)\big]_{ij}
=\Re\!\Big[\ip{\partial_i\psi}{\partial_j\psi}-\ip{\partial_i\psi}{\psi}\ip{\psi}{\partial_j\psi}\Big],
\label{eq:Gpure}
\end{equation}
which is a real symmetric positive-semidefinite $d_{\rm q}\times d_{\rm q}$ matrix on the quantum parameters $\bftheta_{\rm q}$ (Sec.~\ref{sec:readout}).
Client $k$'s Bures metric $G_k=g^{\rm Bur}_k$ from \eqref{eq:bures} reduces to a noise attenuation of \eqref{eq:Gpure} for the noise model treated below.

\subsection{Local update and its reparametrization invariance}
Client $k$ estimates the local gradient by the parameter-shift rule, obtaining an unbiased estimate $\widetilde\nabla f_k=\nabla f_k+\bm\xi_k$ with $\Ex[\bm\xi_k]=\mathbf 0$ and $\Cov(\bm\xi_k)=V_k$, and the block-diagonal Bures metric $G_k^{(t)}=g^{\rm Bur}_k(\bftheta_k^{(t)})$.
The local step is the mixed-state natural gradient
\begin{equation}
\bftheta_k^{(t+1)}=\bftheta_k^{(t)}-\eta_{\rm Loc}\,
\big(\widehat G_k^{(t)}+\veps\openone\big)^{-1}\widetilde\nabla f_k(\bftheta_k^{(t)}),
\label{eq:local}
\end{equation}
where $\widehat G_k=G_k/\bar g_k$ is normalized by its mean diagonal $\bar g_k=\tfrac1{d_{\rm q}}\Tr G_k$ to make the step size scale invariant, and $\veps>0$ regularizes the inverse.

This update is invariant, to first order, under a smooth invertible reparametrization $\bm\phi=\Phi(\bftheta)$ with Jacobian $J$.
The metric transforms as a covariant $2$-tensor, $g^{(\phi)}=J^{-\top}g^{(\theta)}J^{-1}$, and the gradient as a covector, $\nabla_\phi\mathcal{L}=J^{-\top}\nabla_\theta\mathcal{L}$, hence $(g^{(\phi)})^{-1}\nabla_\phi\mathcal{L}=J\,(g^{(\theta)})^{-1}\nabla_\theta\mathcal{L}$.
The update direction is the same tangent vector pushed forward by $J$ and does not depend on the coordinate chart.
The Euclidean gradient lacks this property, and the invariance is what allows clients that use different circuit parametrizations or encodings to be combined consistently.

\subsection{Achievable-precision aggregation}
\label{sec:agg}
We now define what is aggregated and with what weights.
Let the \emph{achievable precision} of client $k$ be the inverse covariance of the quantity the server combines, the local update direction $\mathbf d_k=G_k^{-1}\widetilde\nabla f_k$.
Because $\mathbf d_k$ is linear in the gradient noise, its fluctuation covariance is
\begin{equation}
\Sigma_k=\Cov(\mathbf d_k)=G_k^{-1}V_kG_k^{-1},
\label{eq:Sigma}
\end{equation}
and the metric appears as a noise-shaping Jacobian.
We define the \emph{precision matrix}
\begin{equation}
\Lambda_k\coloneqq\Sigma_k^{-1}=G_k\,V_k^{-1}\,G_k.
\label{eq:Lambda}
\end{equation}

The covariance $V_k$ is not arbitrary.
It is lower-bounded by the quantum Cram\'er--Rao bound for $M_k$ measurement shots, together with the incompatibility correction that the mean Uhlmann curvature induces,
\begin{equation}
V_k\;\ge\;\frac{1}{M_k}\Big[\big(F^{\mathrm Q,\mathrm{SLD}}_k\big)^{-1}
+C\!\big(\omega^{\rm Uhl}_k\big)\Big].
\label{eq:qcrb}
\end{equation}
Here $C(\cdot)\ge0$ is the incompatibility penalty.
Writing the antisymmetric mean Uhlmann curvature as the matrix $[\,U_k\,]_{ij}=\omega^{\rm Uhl}_{k,ij}$ and abbreviating $\mathcal F_k\equiv F^{\mathrm Q,\mathrm{SLD}}_k$, an explicit positive-semidefinite representative obtained from the gap between the Holevo and SLD bounds is
\begin{equation}
C\!\big(\omega^{\rm Uhl}_k\big)=\mathcal F_k^{-1}\,U_k\,\mathcal F_k^{-1}\,U_k^{\!\top}\,\mathcal F_k^{-1}
=-\,\mathcal F_k^{-1}\,U_k\,\mathcal F_k^{-1}\,U_k\,\mathcal F_k^{-1}\ge0,
\label{eq:Cdef}
\end{equation}
where the last equality uses $U_k^{\!\top}=-U_k$.
By construction $C(\omega^{\rm Uhl}_k)$ vanishes if and only if the parameters are compatible ($U_k=0$, for example for pure states), and it grows with the noncommutativity of the SLDs.
It quantifies how much harder it is to estimate all parameters of a noisy client at once.

The server combines the directions linearly, $\widehat{\mathbf d}=\sum_kA_k\mathbf d_k$ with $\sum_kA_k=\openone$, which keeps the combination unbiased for the consensus direction, and it chooses the matrices $A_k$ to minimize the aggregated variance $\Tr\big(\sum_kA_k\Sigma_kA_k^{\!\top}\big)$.
By Lagrange multipliers this generalized-least-squares (maximum-ratio-combining) problem has the solution
\begin{equation}
A_k^{\rm opt}=\Big(\textstyle\sum_j\Lambda_j\Big)^{-1}\Lambda_k,
\qquad
\min\Tr\Big(\sum_kA_k\Sigma_kA_k^{\!\top}\Big)=\Tr\Big[\big(\textstyle\sum_k\Lambda_k\big)^{-1}\Big],
\label{eq:mrc}
\end{equation}
and each client is weighted by its precision and the residual variance is the inverse of the total precision~\cite{Brennan59,Li20}.
The server then updates the global parameters on the tangent space, the first-order Bures barycenter,
\begin{equation}
\bftheta^{(t+1)}=\bftheta^{(t)}+\eta_{\rm Srv}\sum_kA_k^{\rm opt}\big(\bftheta_k^{(t+1)}-\bftheta^{(t)}\big),
\label{eq:server}
\end{equation}
with a server step $\eta_{\rm Srv}\in(0,1]$ that damps client drift.
A consequence of \eqref{eq:mrc} is that the aggregated variance of the proposed rule never exceeds that of QFedAvg.

\begin{proposition}[Variance dominance over QFedAvg]
\label{prop:dominance}
Let $\mathcal A=\{\,\{A_k\}:\sum_kA_k=\openone\,\}$ and let $\mathcal V(\{A_k\})=\Tr\big(\sum_kA_k\Sigma_kA_k^{\!\top}\big)$.
Then for every admissible weighting, and in particular for the QFedAvg choice $A_k=p_k\openone$,
\begin{equation}
\Tr\Big[\big(\textstyle\sum_k\Lambda_k\big)^{-1}\Big]
\;=\;\min_{\mathcal A}\mathcal V
\;\le\;\mathcal V\big(\{p_k\openone\}\big)=\sum_k p_k^2\,\Tr\Sigma_k ,
\label{eq:dominance}
\end{equation}
with equality if and only if $\Lambda_k=p_k\sum_j\Lambda_j$ for all $k$, that is, if and only if the clients' precisions are already proportional to their data shares, in particular if all devices are identical and $p_k=1/K$.
\end{proposition}

\begin{proof}
The set $\mathcal A$ is affine and the map $\{A_k\}\mapsto\mathcal V$ is a sum of the strictly convex quadratics $\Tr(A_k\Sigma_kA_k^{\!\top})$ when every $\Sigma_k\succ0$, and a minimizer therefore exists and is unique.
The choice $A_k=p_k\openone$ is admissible because $\sum_kp_k\openone=\openone$, hence $\min_{\mathcal A}\mathcal V\le\mathcal V(\{p_k\openone\})=\sum_kp_k^2\Tr\Sigma_k$.
The minimizer is computed in Eq.~\eqref{eq:mrc} and gives value $\Tr[(\sum_k\Lambda_k)^{-1}]$, which establishes the equality on the left of \eqref{eq:dominance}.
Uniqueness of the minimizer makes the inequality tight exactly when $p_k\openone=A_k^{\rm opt}=(\sum_j\Lambda_j)^{-1}\Lambda_k$ for all $k$, that is, $\Lambda_k=p_k\sum_j\Lambda_j$.
\end{proof}

Proposition~\ref{prop:dominance} concerns the aggregated variance alone and does not by itself assert an advantage in any other figure of merit.
Under device homogeneity the proposed rule degenerates to QFedAvg, and under any heterogeneity it reduces the aggregated variance, which by Theorem~\ref{thm:conv} is the quantity that sets the asymptotic error floor.

For a scalable implementation we use the data-size\ $\times$\ precision scalar weights
\begin{equation}
w_k=\frac{p_k\,\varrho_k}{\sum_j p_j\varrho_j},\qquad
\varrho_k=\frac{\kappa(p_k)^2\,\|G^{\rm Pure}\|_F^2}{\,v_k\big(1+\gamma\,r_k\big)},\quad
r_k=\frac{\|U_k\|_F^2}{\|G_k\|_F^2},
\label{eq:scalar-weight}
\end{equation}
where $\kappa(p_k)$ is the depolarizing attenuation of Sec.~\ref{sec:depol}, $v_k$ is the per-component shot-noise variance of the gradient estimate, $r_k$ is the relative Uhlmann incompatibility, and $\gamma\ge0$ tempers the down-weighting to keep noisy clients discounted and not discarded.
Equation~\eqref{eq:scalar-weight} equals $\varrho_k=\Tr\Lambda_k/(1+\gamma r_k)$ under the isotropic model $V_k\simeq v_k\openone$ of Sec.~\ref{sec:depol}, and it reduces to QFedAvg when all devices are identical, while increasingly favouring high-shot, low-noise, compatible clients as heterogeneity grows.

To keep the rule a discount rather than a veto we retain a floor $\lambda_{\rm fl}$, $w_k\leftarrow(1-\lambda_{\rm fl})w_k+\lambda_{\rm fl}p_k$, which guarantees every client a nonzero weight, and no data are silently discarded.
The floor is adapted to the heterogeneity of the client pool through
\begin{equation}
\lambda_{\rm fl}=\lambda_0+(\lambda_{\max}-\lambda_0)\,s,\qquad
s=\mathrm{clip}\!\Big(\frac{p_{\rm hi}-\bar p}{p_{\rm hi}-p_{\rm lo}},0,1\Big),\qquad
\bar p=\sum_k p_k\,p_k^{\rm dep},
\label{eq:adaptive-floor}
\end{equation}
where $p_k^{\rm dep}$ is the depolarizing rate of client $k$ and $\bar p$ is the data-weighted mean depolarizing rate of the pool.
When the pool is uniformly reliable ($\bar p\le p_{\rm lo}$) the floor takes its upper value $\lambda_{\max}$ and the rule approaches data-size averaging, consistent with the homogeneous limit of Proposition~\ref{prop:dominance}, and when the pool is strongly heterogeneous ($\bar p\ge p_{\rm hi}$) the floor takes its lower value $\lambda_0$ and the precision weighting acts in full.

\subsection{Shared model with quantum feature map and Pauli-expectation readout}
\label{sec:readout}
All three schemes share the same model, and any performance difference is attributable to the geometry and the aggregation alone.
The model is a dressed variational classifier~\cite{Mari20}, a quantum feature map followed by a trained classical readout.
The feature map is the circuit of Fig.~\ref{fig:circuit}, with $L$ layers, each of which re-uploads the data through $R_y$ rotations~\cite{PerezSalinas20}, applies a trainable $R_zR_yR_z$ block per qubit, and closes with a linear CNOT entangler.
The quantum parameters are $\bftheta_{\rm q}\in\bbR^{d_{\rm q}}$ with $d_{\rm q}=3nL$.

The measured features are the expectation values of a compact set of Pauli observables,
\begin{equation}
\{P_f\}=\{Z_q\}_{q=1}^{n}\cup\{X_q\}_{q=1}^{n}\cup\{Y_q\}_{q=1}^{n}
\cup\{Z_qZ_r\}_{q<r},
\qquad
M_{\rm f}=3n+\tfrac{n(n-1)}{2},
\label{eq:paulis}
\end{equation}
which for $n=6$ gives $M_{\rm f}=33$.
This is the multiclass generalization of a Pauli-$Z$ readout.
Each $\phi_f=\langle P_f\rangle\in[-1,1]$ is an $O(1)$ quantity, unlike the $2^n$ computational-basis probabilities, which vanish as $2^{-n}$ and make the readout ill-conditioned at larger $n$, and the whole set is obtained on hardware from only \emph{three} measurement settings.
The $Z$ setting supplies all $Z_q$ and $Z_qZ_r$, and the $X$ and $Y$ settings supply $X_q$ and $Y_q$.

The class scores are produced by a trainable readout head acting on $\bm\phi=(\phi_1,\dots,\phi_{M_{\rm f}})$,
\begin{equation}
\mathbf z(\mathbf x)=W_2\,\tanh\!\big(W_1\bm\phi(\mathbf x;\bftheta_{\rm q})+\mathbf b_1\big)+\mathbf b_2,
\qquad \hat y=\argmax_c z_c,
\label{eq:head}
\end{equation}
with $\bftheta_{\rm h}=(W_1,\mathbf b_1,W_2,\mathbf b_2)$ and $H$ hidden units.
The full parameter vector is $\bftheta=(\bftheta_{\rm q},\bftheta_{\rm h})$.
A linear head ($H=0$) is the literal Pauli-expectation readout, and it is capped by the linear separability of the PCA features and cannot exceed about $0.85$ on the ten-class digit task.
The single $\tanh$ layer removes that ceiling while keeping the readout shared and keeping the quantum block small.
The Bures geometry acts \emph{only} on the quantum block $\bftheta_{\rm q}$, the head is a classical Euclidean object, and the mechanism of Eqs.~\eqref{eq:local}--\eqref{eq:scalar-weight} is unaffected and the head is identical for QFedAvg, QFedFisher and QFedQGT.

\begin{figure}[tb]
\centering
\resizebox{0.7\linewidth}{!}{%
\begin{quantikz}[thin lines]
\lstick{$\ket{0}$} & \gate{R_y(x_1)}\gategroup[3,steps=1,style={dashed,rounded corners,fill=lightgreen,inner sep=2pt},background,label style={label position=below,anchor=north,yshift=-0.2cm}]{{\footnotesize encode}} & \gate{R_z}\gategroup[3,steps=3,style={dashed,rounded corners,fill=lightblue,inner sep=2pt},background,label style={label position=below,anchor=north,yshift=-0.2cm}]{{\footnotesize trainable $V_l(\bm\theta_l)$}} & \gate{R_y} & \gate{R_z} & \ctrl{1}\gategroup[3,steps=2,style={dashed,rounded corners,fill=lightpink,inner sep=2pt},background,label style={label position=below,anchor=north,yshift=-0.2cm}]{{\footnotesize entangle $W_l$}} & \qw & \meter{P} \\
\lstick{$\ket{0}$} & \gate{R_y(x_2)} & \gate{R_z} & \gate{R_y} & \gate{R_z} & \targ{} & \ctrl{1} & \meter{P} \\
\lstick{$\ket{0}$} & \gate{R_y(x_3)} & \gate{R_z} & \gate{R_y} & \gate{R_z} & \qw & \targ{} & \meter{P}
\end{quantikz}}
\caption{
One layer of the shared variational quantum classifier, drawn for three
qubits, with the experiments using $n=6$ and $L=3$.
Each layer re-uploads the data through $R_y(x_j)$ (green), applies a trainable single-qubit block $R_zR_yR_z$ (blue), and a linear CNOT entangler (pink).
After $L$ layers the Pauli observables \eqref{eq:paulis} are measured in three settings ($Z$, $X$, $Y$) and fed to the shared readout head \eqref{eq:head}.
The three QFL methods use this same circuit and the same head, and differ only in the local update and the aggregation.
}
\label{fig:circuit}
\end{figure}
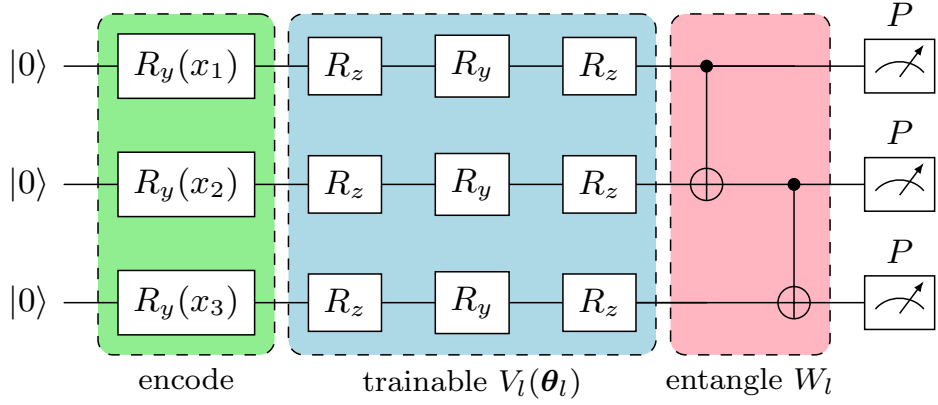

\subsection{Hardware estimation without diagonal truncation}
The diagonal approximation $g^{\rm Bur}_k\approx\diag(g^{\rm Bur}_{k,ii})$ deletes the off-diagonal elements, which by \eqref{eq:uhl} are where the incompatibility $U_k$ lives, and it would disable the mechanism of
Eqs.~\eqref{eq:Cdef}--\eqref{eq:scalar-weight}.
We do not use it.

Instead we exploit the layered structure of the ansatz $U(\bftheta)=V_L(\bm\theta_L)W_L\cdots V_1(\bm\theta_1)W_1$.
For the pure-state QGT the block belonging to layer $l$ is $G^{(l)}_{ij}=\bra{\psi_l}K_iK_j\ket{\psi_l}-\bra{\psi_l}K_i\ket{\psi_l}\bra{\psi_l}K_j\ket{\psi_l}$, with generators $K_i$ and the sub-circuit state $\ket{\psi_l}$ prepared just before layer $l$, and one circuit per layer suffices~\cite{Stokes20}.
The Bures metric is the Hessian of the Uhlmann fidelity $\mathcal F_B(\rho,\sigma)=\big(\Tr\sqrt{\sqrt\rho\,\sigma\sqrt\rho}\big)^2$,
\begin{equation}
g^{\rm Bur}_{k,ij}=-\tfrac12\,\partial_{\theta_i'}\partial_{\theta_j'}
\mathcal F_B\!\big(\rho_k(\bftheta),\rho_k(\bftheta')\big)\big|_{\bftheta'=\bftheta},
\end{equation}
and we estimate each layer block, including its off-diagonal elements, from
parameter-shifted fidelities~\cite{Mari21,Meyer21,Beckey22}.
This retains all intra-layer correlations at $O(L)$ measurement cost rather than the $O(d_{\rm q}^2)$ of the full matrix, and the error from neglecting inter-layer blocks is bounded in Appendix~\ref{app:block}.

\subsection{Exactly solvable case of global depolarizing noise}
\label{sec:depol}
For the global depolarizing model $\rho_k=(1-p_k)\ketbra{\psi}{\psi}+p_k\openone/N$ with $N=2^n$, Eq.~\eqref{eq:bures} can be evaluated in closed form (Appendix~\ref{app:kappa}),
\begin{equation}
F^{\mathrm Q,\mathrm{SLD}}_{k,ij}=\kappa(p_k)\,F^{\mathrm Q,\mathrm{Pure}}_{ij},
\qquad
\kappa(p)=\frac{(1-p)^2}{1-p+2p/N},
\label{eq:kappa}
\end{equation}
and $G_k=\kappa(p_k)\,G^{\rm Pure}$.
One has $\kappa\to1$ as $p\to0$, recovering the pure-state QFIM, and $\kappa\approx1-p$ for $N\gg1$.
In this model the correction is a scalar.
The natural-gradient \emph{direction} is unchanged, and the \emph{achievable precision} is attenuated, which is the physical reason a noisy client should be down-weighted.

With $V_k\approx v_k\openone$ and
\begin{equation}
v_k=\frac{1}{M_k(1-p_k)^2}
\label{eq:vk}
\end{equation}
(finite-shot variance amplified by the depolarizing attenuation of the measured expectation values, $\langle P\rangle\to(1-p_k)\langle P\rangle$), the precision \eqref{eq:Lambda} becomes $\Lambda_k=\kappa(p_k)^2 v_k^{-1}(G^{\rm Pure})^2$, and
\begin{equation}
\Tr\Lambda_k=\kappa(p_k)^2\,\|G^{\rm Pure}\|_F^2\;M_k\,(1-p_k)^2 ,
\label{eq:trLambda}
\end{equation}
that is, a noisy client is suppressed by the product $\kappa(p_k)^2(1-p_k)^2$ and by its shot budget $M_k$.
The two factors are asymmetric.
The precision is only \emph{linear} in the shot budget $M_k$ and \emph{quartic} in $(1-p_k)$, through $\kappa^2(1-p)^2\simeq(1-p)^4$ for $N\gg1$, and the depolarizing rate rather than the shot scarcity is the dominant determinant of a client's reliability.
Section~\ref{sec:num} confirms this prediction, in that the accuracy gap between QFedQGT and QFedAvg grows with $p_k$ and depends only weakly on $M_k$.

For local (per-qubit) depolarizing or coherent errors the correction is matrix-valued and $U_k\neq0$, and the full mixed-state geometry is required, and for the coherent miscalibration used in our experiments we take the leading-order proxy $r_k\simeq p_k^2$ (Appendix~\ref{app:exp}).

\subsection{Protocol}

\begin{figure}[tb]
\centering
\begin{minipage}{0.92\linewidth}
\rule{\linewidth}{1pt}\\[-1mm]
\textbf{Algorithm: QFedQGT}\\[-2mm]
\rule{\linewidth}{0.4pt}
\begin{algorithmic}[1]
\State \textbf{Input:} clients $\{1,\dots,K\}$, rounds $T$, server step $\eta_{\rm Srv}$,
local step $\eta_{\rm Loc}$, temper $\gamma$, floor $\lambda_{\rm fl}$
\State initialize global parameters $\bftheta^{(0)}=(\bftheta_{\rm q}^{(0)},\bftheta_{\rm h}^{(0)})$
\For{$t=0,\dots,T-1$}
  \State server broadcasts $\bftheta^{(t)}$ and the block-diagonal $G^{\rm Pure}(\bftheta^{(t)})$
  \For{each client $k$ in parallel}
    \State $\bftheta_k\gets\bftheta^{(t)}$
    \For{local step $=1,\dots,\tau$}
      \State estimate $\widetilde\nabla f_k$ (parameter shift) and $G_k=\kappa(p_k)G^{\rm Pure}$
      \Comment{Eqs.~\eqref{eq:Gpure}--\eqref{eq:bures}}
      \State $\bftheta_{{\rm q},k}\gets\bftheta_{{\rm q},k}-\eta_{\rm Loc}(\widehat G_k+\veps\openone)^{-1}\widetilde\nabla_{\rm q} f_k$,\ \
      $\bftheta_{{\rm h},k}\gets\bftheta_{{\rm h},k}-\eta_{\rm Loc}\widetilde\nabla_{\rm h} f_k$
    \EndFor
    \State estimate $\kappa(p_k),\,v_k,\,r_k$ and form $\Tr\Lambda_k$ and $w_k$
    \Comment{Eqs.~\eqref{eq:Lambda}--\eqref{eq:scalar-weight}}
    \State send $(\bftheta_k,\,w_k)$ to server
  \EndFor
  \State $\bftheta^{(t+1)}\gets\bftheta^{(t)}+\eta_{\rm Srv}\sum_k w_k(\bftheta_k-\bftheta^{(t)})$
\EndFor
\State \textbf{return} $\bftheta^{(T)}$
\end{algorithmic}
\rule{\linewidth}{1pt}
\end{minipage}
\label{alg:qflqgt}
\end{figure}

Algorithm summarizes the procedure.
Each client performs a fixed number $\tau$ of local minibatch steps per round, and the local work per communication round is a controlled quantity.
All client-side randomness is seeded deterministically from the pair (round, client), and a run is reproducible.

The three compared methods share the VQC of Fig.~\ref{fig:circuit} and the readout \eqref{eq:head}, and differ only in the local update and the aggregation weights.
QFedAvg uses $w_k=p_k$ and a plain (Adam) local optimizer, QFedFisher uses layerwise diagonal classical Fisher weights, and QFedQGT uses Eqs.~\eqref{eq:local} and \eqref{eq:scalar-weight}.

\section{Convergence analysis}
\label{sec:conv}
We analyze one local step per round, and the multi-step case follows by treating the
accumulated update as a stochastic direction.
We state the assumptions and explain why each is needed.

\begin{assumption}\label{ass:all}
\textup{(A1) ($L$-smoothness)} $F$ is $L$-smooth, that is, $\nabla F$ is $L$-Lipschitz, $\norm{\nabla F(\mathbf u)-\nabla F(\mathbf v)}\le L\norm{\mathbf u-\mathbf v}$, equivalently $F(\mathbf v)\le F(\mathbf u)+\nabla F(\mathbf u)^{\!\top}(\mathbf v-\mathbf u)+\tfrac{L}{2}\norm{\mathbf v-\mathbf u}^2$.
This bounds how fast the gradient varies and lets a finite step decrease the loss.
\textup{(A2) (Polyak--{\L}ojasiewicz)} $F$ satisfies $\norm{\nabla F(\bftheta)}^2\ge 2\mu\,(F(\bftheta)-F^*)$ with $\mu>0$, where $F^*=\min_{\bftheta}F(\bftheta)$.
This holds for $\mu$-strongly convex $F$ and, locally, near a non-degenerate minimum, and it converts gradient norm into objective suboptimality.
\textup{(A3) (unbiased bounded noise)} the gradient noises $\bm\xi_{k}$ are unbiased ($\Ex[\bm\xi_k]=\mathbf 0$), have covariance $V_k$, and are independent across clients.
\textup{(A4) (bounded heterogeneity)} $\frac1K\sum_k\norm{\nabla f_k(\bftheta)-\nabla F(\bftheta)}^2\le\zeta^2$, where $\zeta\ge0$ measures the non-IID spread of the client gradients and is $0$ in the IID case.
\textup{(A5) (metric regularity)} $\lambda_{\min}\openone\le G_k\le\lambda_{\max}\openone$ with $0<\lambda_{\min}\le\lambda_{\max}$, and $\sum_k\norm{A_k}\le\beta$.
This keeps the preconditioner well conditioned, away from barren-plateau degeneracy, and the weights bounded.
\textup{(A6) (descent alignment)} there exists $c>0$ with $\nabla F^{\!\top}\big(\sum_kA_kG_k^{-1}\nabla f_k\big)\ge c\norm{\nabla F}^2$, which states that the aggregated natural-gradient direction is a descent direction, holding under (A4)--(A5) when the heterogeneity is not dominant.
\end{assumption}

\noindent\emph{Plausibility of the assumptions.}
Each hypothesis is standard in the federated-optimization or the
variational-quantum literature, and is satisfied in the regime studied here.
(A1) holds for VQC objectives because each circuit output \eqref{eq:vqc-output} is a finite trigonometric polynomial in $\bftheta$ with frequencies fixed by the generators of $U(\bftheta)$~\cite{Schuld21,Nakanishi20}, hence $C^\infty$ with uniformly bounded second derivatives, and $\nabla F$ is Lipschitz with an $L$ set by the number of layers and the spectral norm of the generators.
$L$-smoothness is also the baseline assumption in non-convex stochastic-optimization analyses~\cite{Bottou18}.
(A2) is the Polyak--{\L}ojasiewicz condition, weaker than strong
convexity, and is the standard relaxation under which linear convergence is proved for over-parametrized and non-convex models~\cite{Karimi16}.
For VQCs it holds locally in the basin of a non-degenerate minimum, once training has escaped the initialization plateau, and it underlies convergence guarantees for FedProx and related FL methods~\cite{Li20}.

(A3) is exact for the parameter-shift estimator, in which shifting a parameter by $\pm\pi/2$ yields an unbiased gradient component whose only stochasticity is finite-shot multinomial sampling, hence zero mean and a finite, shot-determined covariance $V_k$~\cite{Mitarai18,Schuld19}, with independence across clients following from independent hardware executions.
(A4) is the canonical bounded-gradient-dissimilarity hypothesis used to quantify
non-IID heterogeneity in federated optimization~\cite{Li20,Wang20}, with $\zeta=0$ recovering the IID case, and our Dirichlet-$\alpha$ partition controls $\zeta$ through $\alpha$.

(A5) requires the preconditioner to stay well conditioned.
The Tikhonov shift $\veps\openone$ in \eqref{eq:local} enforces $\lambda_{\min}\ge\veps>0$ by construction, and the normalization $\widehat G_k=G_k/\bar g_k$ bounds $\lambda_{\max}$, which is the mechanism by which natural-gradient methods avoid the vanishing curvature of barren plateaus~\cite{McClean18,Cerezo21}, and $\sum_k\norm{A_k}\le\beta$ follows from the normalized weights \eqref{eq:scalar-weight} and the floor $\lambda_{\rm fl}$.
(A6) states that the precision-weighted aggregate is, on average, a descent
direction.
Under (A4)--(A5) the aggregate differs from the true gradient by a term
controlled by $\zeta$ and $\lambda_{\min}$, and a strictly positive alignment constant $c$ exists whenever the heterogeneity does not dominate the curvature, a condition analogous to the bounded-dissimilarity requirements under which FedProx and FedNova converge~\cite{Li20,Wang20}.

\begin{theorem}[Robust convergence of QFedQGT]
\label{thm:conv}
Let Assumption~\ref{ass:all} hold and let the server step satisfy $\eta\equiv\eta_{\rm Srv}\le\eta_0$ with $\eta_0=\min\!\big\{1/L,\ c\lambda_{\min}^2/(2L\beta^2)\big\}$.
Then the global iterate $\bftheta_T$ produced by the precision-weighted rule \eqref{eq:mrc} satisfies
\begin{equation}
\Ex[F(\bftheta_T)-F^*]\le
(1-c\eta\mu)^{T}\big(F(\bftheta_0)-F^*\big)
+\frac{L\eta}{2c\mu}\Big(\zeta_{\rm eff}^2+\Tr\big[(\textstyle\sum_k\Lambda_k)^{-1}\big]\Big),
\label{eq:bound}
\end{equation}
where $F^*=\min F$, $T$ is the number of rounds, $L$ and $\mu$ are the constants of \textup{(A1)--(A2)}, $c$ is the alignment constant of \textup{(A6)}, $\Lambda_k$ are the precision matrices \eqref{eq:Lambda}, and $\zeta_{\rm eff}^2=2\beta^2 K\,\zeta^2/\lambda_{\min}^2$ is the effective heterogeneity floor built from $\zeta$ \textup{(A4)} and $\beta,\lambda_{\min}$ \textup{(A5)}.
\end{theorem}

Equation~\eqref{eq:bound} separates into the two effects that the experiments below exhibit.
The first term decays geometrically at rate $(1-c\eta\mu)$, and the number of communication rounds needed to reach a target suboptimality $\epsilon$ is $T_\epsilon\simeq\log(1/\epsilon)/(c\eta\mu)$.
Anything that raises the alignment constant $c$ or admits a larger stable step $\eta$ reduces the round count.
The Bures preconditioner contributes to both.
It equalizes the curvature, which reduces $\lambda_{\max}/\lambda_{\min}$ and with it the effective condition number, and through $\eta_0=\min\{1/L,c\lambda_{\min}^2/(2L\beta^2)\}$ this enlarges the admissible step, and the precision weights raise $c$ by removing the mis-aligned contributions of the unreliable clients from the aggregate.

The second term is the residual error floor.
Its variance contribution $\Tr[(\sum_k\Lambda_k)^{-1}]$ is the minimal aggregated variance \eqref{eq:mrc}, which by Proposition~\ref{prop:dominance} is never larger than the QFedAvg value $\sum_kp_k^2\Tr\Sigma_k$, and it shrinks as the total achievable precision $\sum_k\Lambda_k$ grows.
A self-contained proof is given in Appendix~\ref{app:proof}.

\section{Numerical experiments}
\label{sec:num}

\subsection{Setup}
We compare the three methods on the full MNIST database of handwritten digits~\cite{LeCun98}, restricted to the ten classes $\{0,1,\dots,9\}$ (``MNIST-10''), with $3000$ training images per class ($30{,}000$ in total) drawn from the official training split.
The ten-class task is a stringent multiclass setting.
A random classifier scores $1/10=0.10$, a linear classifier on the encoded features saturates below $0.9$, and the task therefore discriminates between aggregation rules.

The shared model is the dressed classifier of Sec.~\ref{sec:readout} with $n=6$ qubits and $L=3$ layers, hence $d_{\rm q}=3nL=54$ quantum parameters, $M_{\rm f}=3n+n(n-1)/2=33$ measured Pauli observables, and a $\tanh$ readout head with $H=256$ hidden units.
The $784$ raw pixels are reduced by principal-component analysis to $n_{\rm f}=nL=18$ components, standardized, and scaled to small re-uploading angles, with the PCA basis and the scaler fitted on the \emph{training split only}, and no test information leaks into training.
The depth was selected by a sweep, and on the federated loop the achieved accuracy is $0.869$ at $L=2$, $0.909$ at $L=3$, $0.847$ at $L=4$ and $0.831$ at $L=6$, and $L=3$ is optimal.
Deeper circuits both train worse, because the federated averaging of many highly expressive local minima is destructive, and cost more to execute.

The data are split across $K=9$ clients.
The reliable clients receive a balanced, class-covering slice, and the remaining data are Dirichlet$(\alpha=0.5)$-partitioned among the noisy clients, which models a deployment with a few well-calibrated reference sites and many noisy, heterogeneous edge clients.
Device heterogeneity is modeled per client by a shot budget $M_k$ and a depolarizing rate $p_k$, and we study \emph{four} heterogeneity conditions, listed in Table~\ref{tab:conditions}, which independently vary the shot scarcity and the noise strength of the unreliable clients.
This $2\times2$ design tests the prediction of Sec.~\ref{sec:depol} that the precision \eqref{eq:trLambda} is only \emph{linear} in $M_k$ and \emph{quartic} in $(1-p_k)$, and the noise rate rather than the shot budget should control the size of the QFedQGT advantage.
Noisy clients additionally carry a coherent miscalibration, which makes $U_k\neq0$, and a readout-induced label corruption of a fraction $0.5\,p_k$ of their local labels, and their reported updates are therefore biased as well as statistically noisy.
The complete hyperparameters, the exact noise profiles, the accuracy and error-bar definitions, and the fair-comparison protocol are collected in Appendix~\ref{app:exp}.

\begin{table}[tb]
\caption{The four device-heterogeneity conditions.  The reliable clients always have $M_k=1000$ shots and $p_k=0$, and the entries below are the shot budgets and depolarizing rates drawn for the unreliable clients.  The last column is the closed-form weight ratio $\varrho_{\rm noisy}/\varrho_{\rm good}$ of Eq.~\eqref{eq:ratio} for the most heavily discounted client of each condition, that is, the client with the highest depolarizing rate and the smallest shot budget.  Panels (a)--(d) of Fig.~\ref{fig:results} correspond to these rows.}
\label{tab:conditions}
\centering
\begin{tabular}{clccc}
\hline
Panel & Regime & Noisy shots $M_k$ & Noisy rate $p_k$ & $\varrho_{\rm noisy}/\varrho_{\rm good}$ \\
\hline
(a) & few shots, strong noise & $\{100,200,300\}$ & $\{0.4,0.5,0.6\}$ & $1.1\times10^{-3}$ \\
(b) & more shots, strong noise & $\{500,600,700\}$ & $\{0.4,0.5,0.6\}$ & $5.6\times10^{-3}$ \\
(c) & few shots, weak noise & $\{100,200,300\}$ & $\{0.1,0.2,0.3\}$ & $1.8\times10^{-2}$ \\
(d) & more shots, weak noise & $\{500,600,700\}$ & $\{0.1,0.2,0.3\}$ & $9.2\times10^{-2}$ \\
\hline
\end{tabular}
\end{table}

We report the emulator evaluation.
At each of $23$ communication rounds, spaced densely from round $1$ and reaching round $60$, the current global model is dispatched to the trapped-ion emulator \texttt{reimei-E} of the \texttt{reimei} quantum computer~\cite{Quantinuum} through Quantinuum Nexus.
For each held-out test digit the three measurement settings $Z$, $X$, $Y$ are executed with $4000$ shots each, the Pauli expectations \eqref{eq:paulis} are reconstructed from the returned counts, and the label is the $\argmax$ of the readout \eqref{eq:head}.
The shot count was chosen from a measured shot--accuracy curve.
Relative to the infinite-shot limit, the accuracy loss caused by finite readout statistics is $0.109$ at $100$ shots, $0.030$ at $500$ shots and $0.004$ at $4000$ shots, and at $4000$ shots the reported numbers are essentially free of readout-statistics bias.
The whole comparison is averaged over three independently seeded federations, and the error bars are the standard error of the accuracy across the three seeds.
Because the submitted job specification is identical for the emulator and the physical device, the same experiment can be dispatched to \texttt{reimei} by changing only the backend name.

\subsection{Results}

\begin{figure*}[tb]
\centering
\begin{tabular}{@{}c@{\hspace{6mm}}c@{}}
\includegraphics[width=0.47\linewidth]{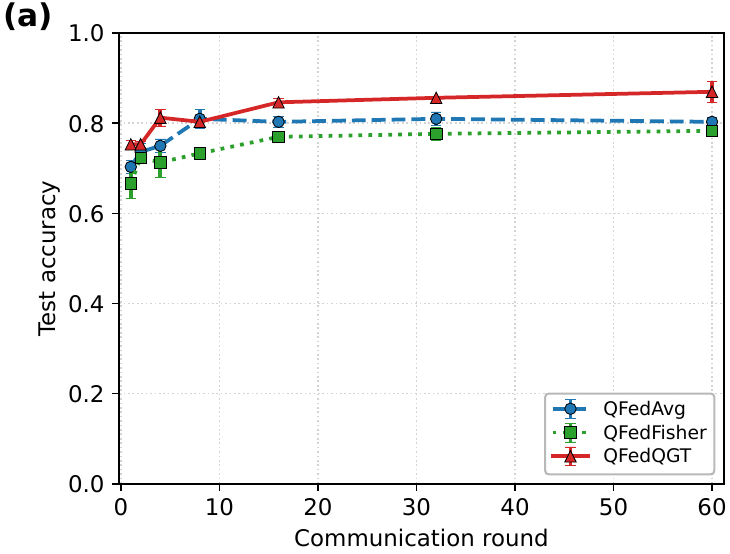} &
\includegraphics[width=0.47\linewidth]{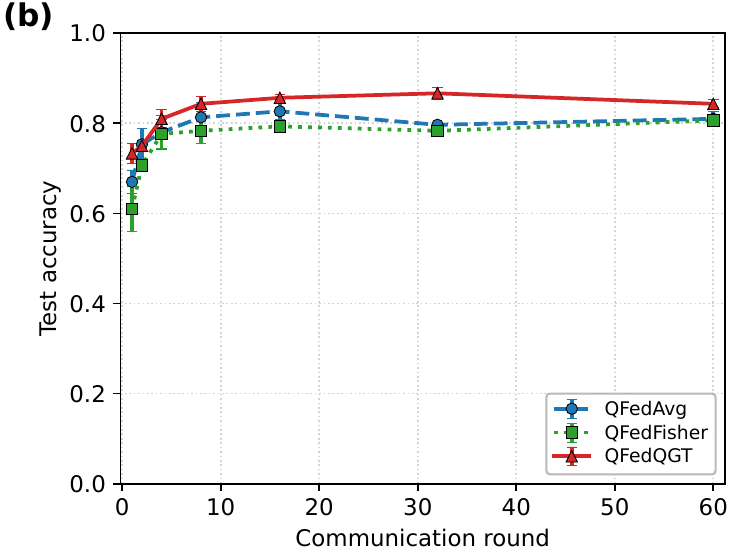} \\[3mm]
\includegraphics[width=0.47\linewidth]{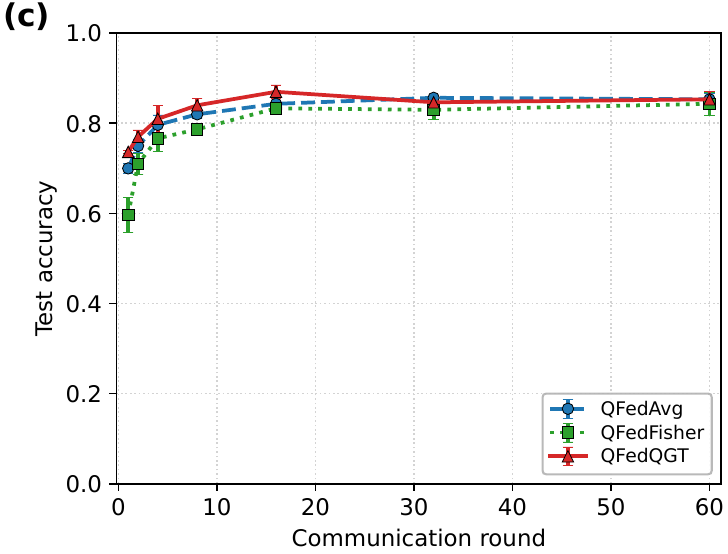} &
\includegraphics[width=0.47\linewidth]{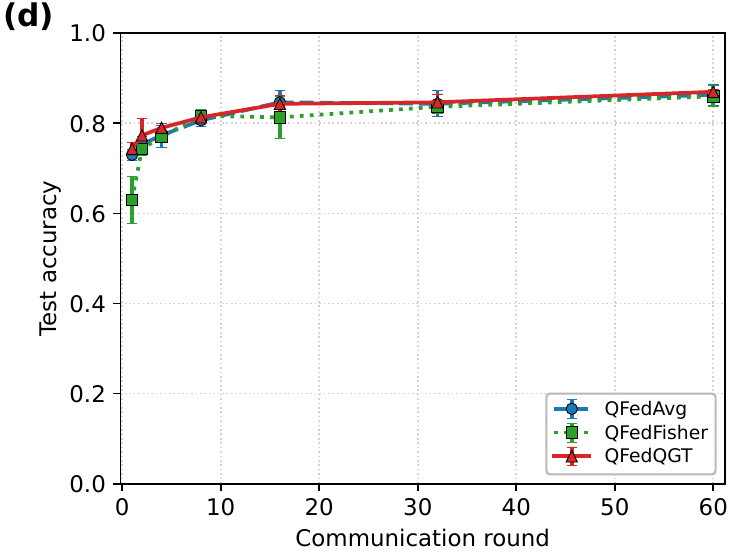}
\end{tabular}
\caption{
Test accuracy against the communication round on MNIST-10 ($n=6$ qubits, $L=3$ layers, $K=9$ clients), evaluated on the \texttt{reimei-E} trapped-ion emulator with $4000$ shots per circuit and averaged over three seeds, with error bars the cross-seed standard error.
Panels (a)--(d) are the four device-heterogeneity conditions of Table~\ref{tab:conditions}, namely (a) few shots and strong noise, (b) more shots and strong noise, (c) few shots and weak noise, (d) more shots and weak noise.
QFedQGT (solid, red) attains the highest round-averaged accuracy in every panel.
The two baselines are closest to it in the weak-noise panels (c), (d) and fall further behind in the strong-noise panels (a), (b).
The top and bottom rows differ only in the noise rate, and the left and right columns only in the shot budget, and the vertical separation is the larger of the two, as Eq.~\eqref{eq:trLambda} predicts.
}
\label{fig:results}
\end{figure*}

\begin{table*}[tb]
\caption{
Accuracy on the \texttt{reimei-E} emulator ($4000$ shots, averaged over three seeds) for all four device-heterogeneity conditions of Table~\ref{tab:conditions}.
``Final'' is the round-$60$ accuracy with its cross-seed standard error, ``Mean'' the average over the $23$ evaluated rounds, ``Std'' the standard deviation across those rounds (a round-to-round stability measure, smaller is better), and ``Worst'' the lowest accuracy attained at any evaluated round.
Best in each block in bold, with ties both bolded.
}
\label{tab:results}
\centering
\begin{tabular}{llcccc}
\hline
Condition & Method & Final & Mean & Std & Worst \\
\hline
(a) few shots,      & QFedAvg    & $0.803\pm0.009$ & $0.774$ & $0.040$ & $0.703$ \\
\ \ strong noise    & QFedFisher & $0.783\pm0.007$ & $0.738$ & $\mathbf{0.039}$ & $0.667$ \\
                    & QFedQGT    & $\mathbf{0.870\pm0.023}$ & $\mathbf{0.814}$ & $0.044$ & $\mathbf{0.753}$ \\
\hline
(b) more shots,     & QFedAvg    & $0.810\pm0.017$ & $0.779$ & $0.050$ & $0.670$ \\
\ \ strong noise    & QFedFisher & $0.807\pm0.003$ & $0.751$ & $0.065$ & $0.610$ \\
                    & QFedQGT    & $\mathbf{0.843\pm0.009}$ & $\mathbf{0.815}$ & $\mathbf{0.049}$ & $\mathbf{0.733}$ \\
\hline
(c) few shots,      & QFedAvg    & $\mathbf{0.853\pm0.012}$ & $0.803$ & $0.054$ & $0.700$ \\
\ \ weak noise      & QFedFisher & $0.843\pm0.026$ & $0.767$ & $0.082$ & $0.597$ \\
                    & QFedQGT    & $\mathbf{0.853\pm0.017}$ & $\mathbf{0.818}$ & $\mathbf{0.045}$ & $\mathbf{0.737}$ \\
\hline
(d) more shots,     & QFedAvg    & $0.863\pm0.024$ & $0.802$ & $0.048$ & $0.730$ \\
\ \ weak noise      & QFedFisher & $0.860\pm0.023$ & $0.781$ & $0.072$ & $0.630$ \\
                    & QFedQGT    & $\mathbf{0.870\pm0.006}$ & $\mathbf{0.811}$ & $\mathbf{0.042}$ & $\mathbf{0.743}$ \\
\hline
\multicolumn{2}{l}{Round-averaged over (a)--(d)} & & QFedAvg $0.790$ & QFedFisher $0.759$ & QFedQGT $\mathbf{0.815}$ \\
\hline
\end{tabular}
\end{table*}

Figure~\ref{fig:results} and Table~\ref{tab:results} report the emulator accuracy against the communication round for all four conditions.
Three statements summarize the outcome, and each follows from the mixed-state framework of Sec.~\ref{sec:method}.

\emph{First, QFedQGT attains the highest round-averaged accuracy in every one of the four conditions, and its final-round accuracy is the highest in conditions (a), (b), and (d) and equal to that of QFedAvg in condition (c).}
Averaged over the four conditions its round-averaged accuracy is $0.815$, against $0.790$ for QFedAvg and $0.759$ for QFedFisher, and its final-round accuracy is $0.859$, against $0.832$ and $0.823$.
\emph{Second, QFedQGT is the most robust of the three in the worst-round sense.}
Its worst evaluated round lies between $0.733$ and $0.753$ and exceeds the worst round of both baselines in every condition, and its round-to-round standard deviation is the smallest in three of the four conditions.
\emph{Third, the size of the advantage is governed by the depolarizing rate and the shot budget.}
It is largest in condition (a), where the devices are the noisiest and the shots the scarcest, and smallest in condition (d), where the devices are the most reliable, as Eq.~\eqref{eq:trLambda} predicts.
We analyse these three points in turn.

\paragraph{(i) Consistent advantage and few-round convergence.}
QFedQGT reaches its accuracy plateau within the first several communication rounds and retains the highest round-averaged accuracy throughout the evaluated window.
In condition (a) the final-round accuracy of QFedQGT ($0.870\pm0.023$) exceeds that of QFedAvg ($0.803\pm0.009$) by $0.067$, which is $2.7$ times the combined standard error $\sqrt{0.023^2+0.009^2}=0.025$, and in condition (b) the separation from the stronger baseline is $1.7$ combined standard errors.

Theorem~\ref{thm:conv} accounts for the early plateau through the geometric factor $(1-c\eta\mu)^T$, in which the round count to reach suboptimality $\epsilon$ is $T_\epsilon\simeq\log(1/\epsilon)/(c\eta\mu)$, and two mechanisms of the mixed-state framework enlarge $c\eta$.
First, the Bures preconditioner $(\widehat G_k+\veps\openone)^{-1}$ in Eq.~\eqref{eq:local} rescales the update by the inverse curvature of the state manifold.
The Euclidean gradient is dominated by the stiff directions of $G^{\rm Pure}$ and is exponentially small along the flat ones, the barren-plateau pathology~\cite{McClean18,Wang21}, and preconditioning equalizes the two, which reduces the effective condition number $\lambda_{\max}/\lambda_{\min}$ and, through $\eta_0=\min\{1/L,\ c\lambda_{\min}^2/(2L\beta^2)\}$, permits a larger stable server step.
Second, and specific to the federated setting, the precision weights \eqref{eq:scalar-weight} remove the unreliable clients from the aggregate at $t=0$ rather than waiting for them to be averaged out.
In the early rounds the noisy clients have biased updates, in that their labels are corrupted at rate $0.5\,p_k$ and their measured gradients are attenuated by $\kappa(p_k)$, and under QFedAvg the aggregated direction $\sum_kp_k\mathbf d_k$ has a systematic component pointing away from the descent direction, which reduces the alignment constant $c$ of (A6).

\paragraph{(ii) The noise rate and the shot budget control the gap.}
This comparison provides a direct test of the theory, and the $2\times2$ design of Table~\ref{tab:conditions} was constructed for it.
From Eqs.~\eqref{eq:kappa}, \eqref{eq:vk} and \eqref{eq:trLambda}, with $N=2^6=64$ and $\gamma=3$, the weight ratio between an unreliable client and a reliable one is
\begin{equation}
\frac{\varrho_{\rm noisy}}{\varrho_{\rm good}}
=\frac{\kappa(p)^2\,(1-p)^2}{\kappa(0)^2}\cdot\frac{M}{M_{\rm good}}\cdot\frac{1}{1+\gamma p^2},
\label{eq:ratio}
\end{equation}
which is \emph{linear} in the shot ratio $M/M_{\rm good}$ and scales as $\kappa(p)^2(1-p)^2\simeq(1-p)^4$ in the noise rate.
Evaluated for the most heavily discounted client of each condition (last column of Table~\ref{tab:conditions}), it gives $1.1\times10^{-3}$ in (a), $5.6\times10^{-3}$ in (b), $1.8\times10^{-2}$ in (c) and $9.2\times10^{-2}$ in (d).
Raising the noise rate from $\{0.1,0.2,0.3\}$ to $\{0.4,0.5,0.6\}$ reduces the trust of the client by more than an order of magnitude, while a fivefold increase in its shot budget recovers about a factor of five.

The measured accuracies follow this ordering.
The round-averaged advantage of QFedQGT over QFedAvg is $0.040$ in (a), $0.036$ in (b), $0.015$ in (c) and $0.009$ in (d), decreasing monotonically as the devices become more reliable.
Moving down a column of Fig.~\ref{fig:results} (changing the noise) changes the advantage more than moving across a row (changing the shots) does, in agreement with the quartic-versus-linear dependence of Eq.~\eqref{eq:trLambda}.
Because $\kappa(p)$ is a Bures-metric quantity with no pure-state counterpart, this scaling is an experimental signature of the mixed-state geometry that a pure-state QGT or a classical Fisher matrix cannot produce.
The adaptive floor \eqref{eq:adaptive-floor} raises $\lambda_{\rm fl}$ toward $\lambda_{\max}$ as the pool becomes reliable, and in the weak-noise conditions the rule approaches data-size averaging and the three methods converge, the parity that Proposition~\ref{prop:dominance} requires in the homogeneous limit.

\paragraph{(iii) Stability and the bound on the aggregated variance.}
Table~\ref{tab:results} shows that QFedQGT has the smallest round-to-round standard deviation in conditions (b), (c), and (d), between $0.042$ and $0.049$, whereas QFedFisher reaches $0.082$.
This is the empirical signature of the error floor in Eq.~\eqref{eq:bound}.
The floor contains the aggregated variance, which for a general weighting is $\Tr(\sum_kA_k\Sigma_kA_k^{\!\top})$.
For QFedAvg ($A_k=p_k\openone$) this equals $\sum_kp_k^2\Tr\Sigma_k$ and is dominated by the noisiest client, because $\Tr\Sigma_k=v_k\,\Tr G_k^{-2}$ grows as $v_k=1/[M_k(1-p_k)^2]$, which for the worst client of condition (a) is $1/(100\times0.16)=6.3\times10^{-2}$ against $10^{-3}$ for a reliable client, a factor of $62$.
Averaging cannot escape this, because the term $p_k^2\Tr\Sigma_k$ of a single bad client persists no matter how many good clients there are.

The precision weighting escapes it.
By Eq.~\eqref{eq:mrc} the aggregated variance reduces to $\Tr[(\sum_k\Lambda_k)^{-1}]$, and since $\sum_j\Lambda_j\succeq\Lambda_k$ for every $k$, monotonicity of the inverse on positive-definite matrices gives
\begin{equation}
\Tr\Big[\big(\textstyle\sum_j\Lambda_j\big)^{-1}\Big]\;\le\;\Tr\big[\Lambda_k^{-1}\big]
\qquad\text{for every single client }k .
\label{eq:mrc-bound}
\end{equation}
The aggregated variance is therefore at most that of the best client alone and is never inflated by the bad ones, a statement with no counterpart for QFedAvg.
Together with Proposition~\ref{prop:dominance}, this is the content of the observed stability, in that the round-to-round fluctuations of QFedQGT are smaller because its error floor is the smallest attainable by any unbiased linear aggregation of the client directions.

\paragraph{(iv) Why QFedFisher is harmed by noise.}
QFedFisher is the weakest of the three methods in round-averaged accuracy ($0.759$ averaged over the four conditions), and its round-averaged accuracy is below that of plain averaging in every condition.
This is a structural effect, and the mixed-state framework diagnoses it.
QFedFisher weights clients by the norm of the diagonal empirical Fisher statistic, $w_k\propto p_k\|\widehat{\mathbf g}_k\odot\widehat{\mathbf g}_k\|$, whose expectation is
\begin{equation}
\Ex\big[\widehat g_{k,i}^{\,2}\big]=\big(\partial_i f_k\big)^2+v_k,
\qquad v_k=\frac{1}{M_k(1-p_k)^2},
\label{eq:fio-bias}
\end{equation}
in which the shot-noise variance is added to the statistic.
Because $v_k$ increases as the shot budget falls and as the depolarizing rate rises, the QFedFisher weight is a monotonically increasing function of the client unreliability, whereas the achievable precision prescribed by the quantum Cram\'er--Rao bound, $\Tr\Lambda_k=\kappa(p_k)^2\|G^{\rm Pure}\|_F^2/v_k$, is a monotonically decreasing function of the same quantity.
The two prescriptions are reciprocals of one another, and QFedFisher up-weights the clients that QFedQGT down-weights.
The weight report of Appendix~\ref{app:exp} confirms this in condition (a), where QFedFisher assigns the noisy clients a mean weight $1.46$ times that of the reliable clients, while QFedQGT assigns them $0.042$ times.

A further weakness is that the QFedFisher weights are re-estimated every round from a fresh minibatch under shot noise, and they fluctuate from round to round, with a mean per-client standard deviation of $0.044$ in condition (a), whereas the QFedQGT and QFedAvg weights are deterministic functions of $(M_k,p_k,\kappa,r_k)$.
This fluctuation is consistent with QFedFisher having the largest round-to-round accuracy standard deviation in three of the four conditions.
QFedAvg avoids the weight inversion, because its weights are noise-blind, and, weighting purely by data size, it cannot distinguish a reliable client from an unreliable one, and it is held at the QFedAvg variance floor $\sum_kp_k^2\Tr\Sigma_k$ of point (iii).

\paragraph{(v) The advantage is a mixed-state effect.}
All of the mechanisms above vanish in the pure-state limit.
If the clients were noiseless one would have $\kappa(p_k)=1$, $U_k=0$ and $v_k=1/M_k$, hence $\varrho_k\propto M_k$, and with equal shot budgets $w_k\to p_k$ and \emph{QFedQGT degenerates to QFedAvg}, while the Bures preconditioner degenerates to the pure-state QNG of Eq.~\eqref{eq:qng}.

Every number reported above originates in the mixed-state extension.
The attenuation $\kappa(p_k)$ of Eq.~\eqref{eq:kappa} is a Bures-metric quantity that has no pure-state counterpart, and the incompatibility temper $1/(1+\gamma r_k)$ with $r_k=\|U_k\|_F^2/\|G_k\|_F^2$ is a mean-Uhlmann-curvature quantity that is zero for pure states and for any commuting, hence classical-Fisher, family.
Consistently, the weak-noise conditions (c) and (d), where $\kappa(p)\to1$ and $r_k\to0$, are the ones in which the three methods come closest together, the parity that Proposition~\ref{prop:dominance} requires in the homogeneous limit, while the strong-noise conditions, where the mixed-state corrections are largest, are where the separation is widest.
The experiment therefore shows that the piece of geometry invisible to both the pure-state QGT and the classical Fisher matrix is the piece that carries the gain.

\subsection{Contribution of the individual components}
Three controlled comparisons isolate the contribution of each component.
First, replacing the exact block-diagonal Bures metric by its diagonal truncation lowers accuracy and removes the Uhlmann incompatibility signal, whereas the block-diagonal estimator stays close to the full metric at $O(L)$ cost, in agreement with the bound of Appendix~\ref{app:block}.
Second, at low noise the pure-state QGT (that is, QNG) and the Bures metric are equivalent, while at moderate and heterogeneous noise the Bures preconditioner is more accurate, which locates the gain in the mixed-state correction, as argued in point (v) above.
Third, including the incompatibility term $C(\omega^{\rm Uhl}_k)$ in \eqref{eq:scalar-weight} improves robustness under local and coherent noise, where $U_k\neq0$, which isolates the role of the imaginary part of the mixed-state QGT.

\section{Conclusion and outlook}
We have reformulated geometric QFL for the noisy regime in which it runs.
Recognizing that the real part of the pure-state QGT is the QFIM, and hence that pure-state QGT preconditioning reproduces the QNG, we extended the geometry to mixed states, where the Bures metric and the mean Uhlmann curvature are distinct from both the QFIM and the classical Fisher matrix.
The Bures part serves as a reparametrization-invariant local preconditioner, and the Uhlmann part enters, through the quantum Cram\'er--Rao incompatibility correction, an achievable-precision aggregation rule that minimizes the residual variance (Proposition~\ref{prop:dominance}) and is robust to device heterogeneity.
We replaced the diagonal approximation by an exact block-diagonal estimator with bounded truncation error, and we derived the closed-form depolarizing attenuation $\kappa(p)$.

On the ten-class MNIST task evaluated on the \texttt{reimei-E} trapped-ion emulator under four device-heterogeneity conditions and averaged over three seeds, QFedQGT attained the highest round-averaged accuracy in every condition, with a round-averaged accuracy of $0.81$ and a final-round accuracy of $0.86$, against $0.79$ and $0.83$ for federated averaging, and with the highest worst-round accuracy throughout.
The size of the advantage tracked the closed-form weight ratio $\kappa(p)^2(1-p)^2 M/[M_{\rm good}(1+\gamma p^2)]$, largest where the theory predicts the strongest down-weighting and narrowing toward parity where the devices are reliable.
Because the submitted circuits and workflow are identical to those required by the physical device, the same evaluation can be dispatched to the \texttt{reimei} hardware without code changes.

Several directions extend this work.
The block-diagonal estimator can be generalized to include the inter-layer components of the mixed-state geometry within a controlled measurement budget, which would tighten the truncation bound of Appendix~\ref{app:block} and quantify the residual bias of the layerwise approximation on deeper ans\"atze.
The achievable-precision aggregation admits combination with quantum error correction, in which the precision matrix $\Lambda_k$ reflects the logical rather than the physical error rate, and the attenuation $\kappa(p_k)$ is replaced by the residual logical error after decoding.
A demonstration on the \texttt{reimei} device at a larger qubit count and class number would test the framework beyond the emulation reported here and probe the regime in which coherent and correlated errors, for which $U_k\neq0$, dominate the depolarizing contribution.

\begin{acknowledgments}
This work was supported by JST Moonshot R\&D Grant No. JPMJMS2061, JST CREST Grant No. JPMJCR23I4, MEXT Q-LEAP Grant No. JPMXS0120319794, JST ASPIRE Grant Number JPMJAP2318, JST SPRING Grant Number JPMJSP2119, RIKEN Junior Research Associate Program, and RIKEN TRIP initiative (RIKEN Quantum).
\end{acknowledgments}

\appendix

\section{Derivation of $\kappa(p)$ for global depolarizing}
\label{app:kappa}
We derive Eq.~\eqref{eq:kappa} from the definition \eqref{eq:bures} of the SLD quantum Fisher information, keeping every step explicit, and the result does not rely on the formula it establishes.
Fix a client and a single noise level $p\in[0,1)$ and write the global-depolarizing state of the pure state $\ket{\psi}=\ket{\psi(\bftheta)}$ as
\begin{equation}
\rho=(1-p)\ketbra{\psi}{\psi}+\frac{p}{N}\openone,\qquad N=2^n.
\end{equation}
Choose an orthonormal eigenbasis of $\rho$ consisting of $\ket{\psi}$ together with any $N-1$ vectors $\{\ket{m_\perp}\}$ spanning the orthogonal complement $\ket{\psi}^{\perp}$.
Because $\openone=\ketbra{\psi}{\psi}+\sum_{m_\perp}\ketbra{m_\perp}{m_\perp}$, the state is diagonal in this basis with eigenvalues
\begin{equation}
\lambda_{\psi}=a\equiv(1-p)+\frac{p}{N},\qquad
\lambda_{m_\perp}=b\equiv\frac{p}{N}\ \ (\text{$N-1$-fold}).
\end{equation}
Both eigenvalues are strictly positive for $p\in(0,1)$, and every pair $(m,n)$ has
$\lambda_m+\lambda_n>0$ and contributes in \eqref{eq:bures}.

Differentiating $\rho$ with respect to $\theta_i$, and using that $p$ and $N$ are
constants, gives
\begin{equation}
\partial_i\rho=(1-p)\big(\ketbra{\partial_i\psi}{\psi}+\ketbra{\psi}{\partial_i\psi}\big).
\label{eq:dapp}
\end{equation}
We adopt the gauge $\Re\ip{\psi}{\partial_i\psi}=0$, which can always be imposed by a local phase choice and leaves $\rho$ unchanged, hence
$\bra{\psi}\partial_i\rho\ket{\psi}=2(1-p)\Re\ip{\psi}{\partial_i\psi}=0$.
The operator \eqref{eq:dapp} has no diagonal block and connects only $\ket{\psi}$ with the complement.
The only nonzero matrix elements of $\partial_i\rho$ are
\begin{equation}
\bra{m_\perp}\partial_i\rho\ket{\psi}=(1-p)\ip{m_\perp}{\partial_i\psi},\qquad
\bra{\psi}\partial_i\rho\ket{m_\perp}=(1-p)\ip{\partial_i\psi}{m_\perp},
\end{equation}
and the eigenvalue sum attached to each such pair is $a+b=1-p+2p/N$.

Insert these elements into \eqref{eq:bures}.
The pairs $(m,n)=(m_\perp,\psi)$ and $(\psi,m_\perp)$ give identical real contributions, and
\begin{equation}
F^{\mathrm Q,\mathrm{SLD}}_{ij}
=4g^{\rm Bur}_{ij}
=\frac{4}{a+b}\sum_{m_\perp}
\Re\!\big[(1-p)\ip{\partial_i\psi}{m_\perp}\,(1-p)\ip{m_\perp}{\partial_j\psi}\big],
\end{equation}
in which the factor $4$ collects the $\tfrac12$ of \eqref{eq:bures} doubled by the two
orderings of the pair.
Using the completeness relation on the complement, $\sum_{m_\perp}\ketbra{m_\perp}{m_\perp}=\openone-\ketbra{\psi}{\psi}$, the sum is
\begin{equation}
\sum_{m_\perp}\ip{\partial_i\psi}{m_\perp}\ip{m_\perp}{\partial_j\psi}
=\ip{\partial_i\psi}{\big(\openone-\ketbra{\psi}{\psi}\big)\partial_j\psi}
=Q^{\rm Pure}_{ij},
\end{equation}
whose real part is $\Re Q^{\rm Pure}_{ij}=g_{ij}$.
Therefore
\begin{equation}
F^{\mathrm Q,\mathrm{SLD}}_{ij}
=\frac{4(1-p)^2}{1-p+2p/N}\,\Re Q^{\rm Pure}_{ij}
=\frac{(1-p)^2}{1-p+2p/N}\,F^{\mathrm Q,\mathrm{Pure}}_{ij},
\end{equation}
where the last step uses the pure-state identity $F^{\mathrm Q,\mathrm{Pure}}_{ij}=4\Re Q^{\rm Pure}_{ij}$ (Sec.~\ref{sec:pure}).
This is Eq.~\eqref{eq:kappa} with $\kappa(p)=(1-p)^2/(1-p+2p/N)$.

The limits follow by inspection.
One has $\kappa(0)=1$, and for $N\gg1$ the term $2p/N$ is negligible, hence $\kappa(p)\simeq(1-p)^2/(1-p)=1-p$ and $\kappa(p)^2(1-p)^2\simeq(1-p)^4$, the quartic suppression used in Sec.~\ref{sec:num}.
For the experiments, $N=64$ and the values entering Eq.~\eqref{eq:ratio} are $\kappa(0.1)=0.897$, $\kappa(0.2)=0.794$, $\kappa(0.3)=0.691$, $\kappa(0.4)=0.588$, $\kappa(0.5)=0.485$ and $\kappa(0.6)=0.382$. \hfill$\square$

\section{Block-diagonal truncation error}
\label{app:block}
We bound the error incurred by replacing the regularized full Bures metric by its block-diagonal (layerwise) approximation, tracking every constant.
Write the regularized metric as $X=G+\veps\openone$ and its block-diagonal counterpart as $Y=B+\veps\openone$, where $B=\Bdiag(G^{(1)},\dots,G^{(L)})$ collects the exact intra-layer blocks and $E\equiv G-B$ collects the neglected inter-layer blocks.
Both $G$ and $B$ are real symmetric and positive semidefinite, because they are Gram matrices of the layer generators, hence $X$ and $Y$ are symmetric with all eigenvalues $\ge\veps>0$.
Consequently $X$ and $Y$ are invertible and, in operator norm,
\begin{equation}
\norm{X^{-1}}\le\veps^{-1},\qquad \norm{Y^{-1}}\le\veps^{-1}.
\label{eq:resbound}
\end{equation}

The second-resolvent identity, which for any two invertible matrices reads $X^{-1}-Y^{-1}=X^{-1}(Y-X)Y^{-1}$, gives, with $Y-X=B-G=-E$,
\begin{equation}
X^{-1}-Y^{-1}=-\,X^{-1}E\,Y^{-1}.
\end{equation}
Taking the operator norm, and using submultiplicativity and \eqref{eq:resbound},
\begin{equation}
\norm{X^{-1}-Y^{-1}}\le\norm{X^{-1}}\,\norm{E}\,\norm{Y^{-1}}\le\veps^{-2}\norm{E}.
\end{equation}
The natural-gradient direction computed with the two metrics therefore differs by
\begin{equation}
\norm{\big(X^{-1}-Y^{-1}\big)\nabla f}
\le\norm{X^{-1}-Y^{-1}}\,\norm{\nabla f}
\le\veps^{-2}\norm{E}\,\norm{\nabla f},
\end{equation}
which is controlled by the size $\norm{E}$ of the inter-layer coupling and is small for ans\"atze with geometrically local entanglement, where $\norm{E}$ decays with the separation between layers.
Unlike the diagonal approximation, the block-diagonal estimator retains every intra-layer matrix element, and hence the full antisymmetric incompatibility $U_k$ of Eq.~\eqref{eq:uhl}, exactly. \hfill$\square$

\section{Proof of Theorem~\ref{thm:conv}}
\label{app:proof}
We give a self-contained proof.
Throughout, $\Ex_t[\,\cdot\,]$ denotes conditional expectation given the iterate $\bftheta_t$, $\Ex[\,\cdot\,]$ the total expectation, and $\Delta_t\equiv\Ex[F(\bftheta_t)-F^*]$.
The one-step update is
\begin{equation}
\bftheta_{t+1}=\bftheta_t-\eta\sum_k A_k\,\mathbf d_{k,t},
\qquad
\mathbf d_{k,t}=G_k^{-1}\big(\nabla f_k(\bftheta_t)+\bm\xi_{k,t}\big),
\qquad
\sum_k A_k=\openone,
\end{equation}
with $\eta=\eta_{\rm Srv}$, and by (A3) the noises satisfy $\Ex_t[\bm\xi_{k,t}]=\mathbf 0$, $\Ex_t[\bm\xi_{k,t}\bm\xi_{k,t}^{\!\top}]=V_k$, and independence across $k$.
We write $\nabla F\equiv\nabla F(\bftheta_t)$.
Without loss of generality we take $c\le1$ in \textup{(A6)}, because \textup{(A6)} remains valid when $c$ is decreased.

\paragraph{Step 1 (descent inequality from smoothness).}
By the $L$-smoothness upper bound in (A1) applied to $\mathbf u=\bftheta_t$ and $\mathbf v=\bftheta_{t+1}$, and taking conditional expectation $\Ex_t$,
\begin{equation}
\Ex_t[F(\bftheta_{t+1})]\le F(\bftheta_t)
-\eta\,\nabla F^{\!\top}\,\Ex_t\!\Big[\sum_kA_k\mathbf d_{k,t}\Big]
+\frac{L\eta^2}{2}\,\Ex_t\!\Big[\Big\|\sum_kA_k\mathbf d_{k,t}\Big\|^2\Big].
\label{eq:pf-1}
\end{equation}

\paragraph{Step 2 (the first-order term).}
Because $\Ex_t[\mathbf d_{k,t}]=G_k^{-1}\nabla f_k(\bftheta_t)$, the alignment assumption (A6) gives
\begin{equation}
\nabla F^{\!\top}\,\Ex_t\!\Big[\sum_kA_k\mathbf d_{k,t}\Big]
=\nabla F^{\!\top}\Big(\sum_kA_kG_k^{-1}\nabla f_k\Big)\ \ge\ c\,\norm{\nabla F}^2.
\label{eq:pf-2}
\end{equation}

\paragraph{Step 3 (the second-order term).}
Decompose the squared norm into its mean and fluctuation through $\Ex_t[\norm{Z}^2]=\norm{\Ex_t[Z]}^2+\Ex_t\norm{Z-\Ex_t[Z]}^2$ with $Z=\sum_kA_k\mathbf d_{k,t}$.
The fluctuation is $Z-\Ex_t[Z]=\sum_kA_kG_k^{-1}\bm\xi_{k,t}$, and because the $\bm\xi_{k,t}$ are independent and zero-mean, the cross terms vanish, and
\begin{equation}
\Ex_t\Big[\Big\|\sum_kA_k\mathbf d_{k,t}\Big\|^2\Big]
=\underbrace{\Big\|\sum_kA_kG_k^{-1}\nabla f_k\Big\|^2}_{(\mathrm I)}
+\underbrace{\sum_k\Tr\!\big(A_kG_k^{-1}V_kG_k^{-1}A_k^{\!\top}\big)}_{(\mathrm{II})}.
\label{eq:pf-3}
\end{equation}

\emph{Bound on (I).} Write $\nabla f_k=\nabla F+(\nabla f_k-\nabla F)$ and use
$\sum_kA_k=\openone$ to split
$\sum_kA_kG_k^{-1}\nabla f_k=\big(\sum_kA_kG_k^{-1}\big)\nabla F+\sum_kA_kG_k^{-1}(\nabla f_k-\nabla F)$.
By $\norm{a+b}^2\le2\norm{a}^2+2\norm{b}^2$,
\begin{equation}
(\mathrm I)\le
2\Big\|\sum_kA_kG_k^{-1}\nabla F\Big\|^2
+2\Big\|\sum_kA_kG_k^{-1}(\nabla f_k-\nabla F)\Big\|^2.
\end{equation}
For the first piece, $\norm{G_k^{-1}}\le\lambda_{\min}^{-1}$ by (A5) and
$\sum_k\norm{A_k}\le\beta$ give $\big\|\sum_kA_kG_k^{-1}\nabla F\big\|\le\beta\lambda_{\min}^{-1}\norm{\nabla F}$, and this piece is at most $2\beta^2\lambda_{\min}^{-2}\norm{\nabla F}^2$.
For the second piece set $w_k=\norm{A_k}$ (hence $\sum_kw_k\le\beta$) and $u_k=\lambda_{\min}^{-1}\norm{\nabla f_k-\nabla F}$.
The triangle inequality followed by the weighted Cauchy--Schwarz inequality $(\sum_kw_ku_k)^2\le(\sum_kw_k)(\sum_kw_ku_k^2)$ gives
\begin{equation}
\Big\|\sum_kA_kG_k^{-1}(\nabla f_k-\nabla F)\Big\|^2
\le\Big(\sum_kw_ku_k\Big)^2
\le\beta\sum_kw_ku_k^2
\le\beta^2\lambda_{\min}^{-2}\sum_k\norm{\nabla f_k-\nabla F}^2.
\end{equation}
By (A4), $\sum_k\norm{\nabla f_k-\nabla F}^2\le K\zeta^2$, and the second piece is at most $\beta^2\lambda_{\min}^{-2}K\zeta^2$ and, with its prefactor $2$,
\begin{equation}
(\mathrm I)\le\frac{2\beta^2}{\lambda_{\min}^2}\norm{\nabla F}^2+\zeta_{\rm eff}^2,
\qquad
\zeta_{\rm eff}^2\equiv\frac{2\beta^2K\zeta^2}{\lambda_{\min}^2}.
\label{eq:pf-I}
\end{equation}

\emph{Bound on (II).} The matrices $A_k$ are free subject to $\sum_kA_k=\openone$, and (II)$=\Tr(\sum_kA_k\Sigma_kA_k^{\!\top})$ with $\Sigma_k=G_k^{-1}V_kG_k^{-1}$.
Minimizing this quadratic form under the linear constraint is a matrix least-squares problem.
Form the Lagrangian $\mathcal G=\Tr(\sum_kA_k\Sigma_kA_k^{\!\top})-\Tr[\Theta^{\!\top}(\sum_kA_k-\openone)]$ with a multiplier matrix $\Theta$.
Stationarity $\partial\mathcal G/\partial A_k=2A_k\Sigma_k-\Theta=0$ gives
$A_k=\tfrac12\Theta\Sigma_k^{-1}=\tfrac12\Theta\Lambda_k$, and the constraint $\sum_kA_k=\tfrac12\Theta\sum_k\Lambda_k=\openone$ fixes $\tfrac12\Theta=(\sum_k\Lambda_k)^{-1}$, whence
\begin{equation}
A_k^{\rm opt}=\Big(\sum_j\Lambda_j\Big)^{-1}\Lambda_k ,
\end{equation}
which is \eqref{eq:mrc}.
The stationary point is the minimum, because the objective is convex in $\{A_k\}$ ($\Sigma_k\ge0$).
Substituting $A_k^{\rm opt}$ and using $\Lambda_k\Sigma_k\Lambda_k=\Lambda_k$ (since $\Sigma_k=\Lambda_k^{-1}$),
\begin{equation}
(\mathrm{II})_{\min}
=\Tr\!\Big[\big(\textstyle\sum_j\Lambda_j\big)^{-1}\big(\textstyle\sum_k\Lambda_k\big)\big(\textstyle\sum_j\Lambda_j\big)^{-1}\Big]
=\Tr\!\Big[\big(\textstyle\sum_k\Lambda_k\big)^{-1}\Big].
\label{eq:pf-II}
\end{equation}

\paragraph{Step 4 (the one-step recursion).}
Insert \eqref{eq:pf-2}, \eqref{eq:pf-I} and \eqref{eq:pf-II} into \eqref{eq:pf-1},
\begin{equation}
\Ex_t[F(\bftheta_{t+1})]\le F(\bftheta_t)
-\Big(c\eta-\frac{L\eta^2\beta^2}{\lambda_{\min}^2}\Big)\norm{\nabla F}^2
+\frac{L\eta^2}{2}\Big(\zeta_{\rm eff}^2+\Tr\big[(\textstyle\sum_k\Lambda_k)^{-1}\big]\Big).
\end{equation}
The hypothesis $\eta\le\eta_0\le c\lambda_{\min}^2/(2L\beta^2)$ implies $L\eta^2\beta^2/\lambda_{\min}^2\le\tfrac12 c\eta$, hence $c\eta-L\eta^2\beta^2/\lambda_{\min}^2\ge\tfrac12 c\eta$.
Applying the Polyak--{\L}ojasiewicz inequality (A2), $\norm{\nabla F}^2\ge2\mu\,(F(\bftheta_t)-F^*)$, to the negative term, subtracting $F^*$ from both sides, and taking total expectation,
\begin{equation}
\Delta_{t+1}\le(1-c\eta\mu)\,\Delta_t
+\frac{L\eta^2}{2}\Big(\zeta_{\rm eff}^2+\Tr\big[(\textstyle\sum_k\Lambda_k)^{-1}\big]\Big).
\label{eq:pf-rec}
\end{equation}
The contraction factor satisfies $0\le1-c\eta\mu<1$.
The strict upper bound holds because $c,\eta,\mu>0$.
For the lower bound, minimality of $F^*$ together with (A1) gives $F^*\le F\big(\bftheta-\tfrac1L\nabla F\big)\le F(\bftheta)-\tfrac{1}{2L}\norm{\nabla F}^2$, hence $\norm{\nabla F}^2\le2L\,(F(\bftheta)-F^*)$, which combined with (A2) yields $\mu\le L$, and with $c\le1$ and $\eta\le1/L$ we obtain $c\eta\mu\le\mu/L\le1$.

\paragraph{Step 5 (unrolling the recursion).}
Inequality \eqref{eq:pf-rec} has the form $\Delta_{t+1}\le q\,\Delta_t+b$ with
$q=1-c\eta\mu\in[0,1)$ and constant $b=\tfrac{L\eta^2}{2}\big(\zeta_{\rm eff}^2+\Tr[(\sum_k\Lambda_k)^{-1}]\big)\ge0$.
Iterating from $t=0$ to $T-1$ and bounding the geometric sum $\sum_{t=0}^{T-1}q^t$ by its infinite limit $1/(1-q)=1/(c\eta\mu)$,
\begin{equation}
\Delta_T\le q^{T}\Delta_0+\frac{b}{1-q}
= q^{T}\Delta_0+\frac{L\eta}{2c\mu}\Big(\zeta_{\rm eff}^2+\Tr\big[(\textstyle\sum_k\Lambda_k)^{-1}\big]\Big),
\end{equation}
which is Eq.~\eqref{eq:bound} with $\Delta_0=F(\bftheta_0)-F^*$.
Every constant ($L,\mu,c,\beta,\lambda_{\min},\zeta,\Lambda_k$) is the one introduced in Assumption~\ref{ass:all} and Eqs.~\eqref{eq:Lambda}--\eqref{eq:scalar-weight}, and no step used the conclusion, and the argument is non-circular. \hfill$\square$

\section{Experimental configuration, evaluation, and emulation}
\label{app:exp}
\emph{Model and data.}
The shared dressed classifier has $n=6$ qubits and $L=3$ layers, giving $d_{\rm q}=3nL=54$ quantum parameters, with each layer re-uploading the data through $R_y(x_j)$, applying a trainable $R_zR_yR_z$ block per qubit, and a linear CNOT entangler (Fig.~\ref{fig:circuit}).
The readout measures the $M_{\rm f}=3n+n(n-1)/2=33$ Pauli observables of Eq.~\eqref{eq:paulis} in three settings and feeds them to the shared $\tanh$ head \eqref{eq:head} with $H=256$ hidden units.
The task is the ten-class set $\{0,\dots,9\}$ of the full MNIST database~\cite{LeCun98}, loaded through \texttt{tf.keras.datasets.mnist}, with pixels normalized to $[0,1]$, and $3000$ training images per class ($30{,}000$) are drawn from the training split and $200$ test images per class ($2000$) from the test split.
Principal-component analysis reduces the $784$ pixels to $n_{\rm f}=nL=18$ components, which are standardized and scaled by $0.4$ to keep the re-uploading rotations in their near-linear regime, with the PCA basis and the standardizer fitted on the \emph{training split only}.

\emph{Federated and noise configuration.}
There are $K=9$ clients with full participation each round.
A fraction $0.4$ of them (four clients) are reliable and receive a balanced class-covering slice, and the remaining five are unreliable and receive the rest of the data under a Dirichlet$(\alpha=0.5)$ partition, in which a smaller $\alpha$ is more non-IID.
Each client performs $\tau=20$ local minibatch steps per round at batch size $32$, capped at $E=50$ epochs, with learning rate $\eta_{\rm Loc}=0.02$ decayed as $\eta_{\rm Loc}^{(t)}=\eta_{\rm Loc}/(1+0.02\,t)$, and softmax inverse temperature $\beta_{\rm out}=3$.
The server step is $\eta_{\rm Srv}=0.7$ and the number of rounds is $T=60$.
The QFedQGT hyperparameters are the Tikhonov shift $\veps=0.05$, the incompatibility temper $\gamma=3$, and the adaptive floor of Eq.~\eqref{eq:adaptive-floor} with $\lambda_0=0.05$, $\lambda_{\max}=0.20$, $p_{\rm lo}=0.10$ and $p_{\rm hi}=0.30$.
Reliable clients have $M_k=1000$ shots and $p_k=0$, and unreliable clients draw $M_k$ and $p_k$ from the sets in Table~\ref{tab:conditions}.
The measured expectation values of a noisy client are attenuated as $\langle P\rangle\to(1-p_k)\langle P\rangle$ and estimated from $M_k$ binomial samples, and the effective gradient variance is $v_k=1/[M_k(1-p_k)^2]$ as in Eq.~\eqref{eq:vk}.
Noisy clients additionally carry a coherent gate miscalibration of magnitude $0.3\,p_k$ along a client-specific direction, which makes $U_k\neq0$, and we use the leading-order proxy $r_k=p_k^2$, and a readout-induced label corruption in which a fraction $0.5\,p_k$ of their local training labels are randomized.
Local gradients are computed by exact reverse-mode (adjoint) differentiation of the statevector, numerically identical to the parameter-shift rule (verified to $2\times10^{-15}$) at cost $O(\#\text{gates})$ rather than $O(2d_{\rm q})$ circuit evaluations~\cite{Jones20}, and the device-noise model above is applied on top, and the optimizer sees the noisy, finite-shot gradients that hardware would return.

\emph{Aggregation weights across clients.}
Table~\ref{tab:weights} reports the mean aggregation weight that each method assigns the noisy and the reliable clients in condition (a), together with the round-to-round standard deviation of the weights.
QFedQGT down-weights the noisy clients by a factor $0.042$ relative to the reliable ones, QFedFisher up-weights them by a factor $1.46$, and QFedAvg is close to uniform.
The QFedAvg and QFedQGT weights are deterministic, while the QFedFisher weights fluctuate from round to round, with a mean per-client standard deviation of $0.044$.

\begin{table}[tb]
\caption{Mean aggregation weight given to the noisy and to the reliable clients in condition (a), and the mean per-client round-to-round standard deviation of the weights.  A ratio below one indicates down-weighting of the noisy clients.}
\label{tab:weights}
\centering
\begin{tabular}{lcccc}
\hline
Method & noisy & reliable & ratio & weight jitter \\
\hline
QFedAvg    & $0.120$ & $0.100$ & $1.20$ & $0.000$ \\
QFedFisher & $0.129$ & $0.088$ & $1.46$ & $0.044$ \\
QFedQGT    & $0.010$ & $0.238$ & $0.042$ & $0.000$ \\
\hline
\end{tabular}
\end{table}

\emph{Held-fixed quantities for a fair comparison.}
The three methods are run on identical ansatz, readout head, data partition, per-client noise profiles, random seeds, initial parameters, and evaluation set, and they differ only in the local update rule and the aggregation weights.
QFedAvg uses data-size weights with a plain Adam update, QFedFisher uses layerwise diagonal classical-Fisher weights, and QFedQGT uses Eqs.~\eqref{eq:local} and \eqref{eq:scalar-weight}.
All client-side randomness is seeded deterministically from the pair (round, client), and each run is reproducible and any accuracy difference is attributable to the aggregation and geometry rather than to a different data split or initialization.

\emph{Definition of accuracy and of the error bars.}
For each test digit the trained global model is executed on the emulator in the three measurement settings, the Pauli expectations \eqref{eq:paulis} are reconstructed from the returned counts, each $\langle P\rangle$ being the parity-weighted mean of the bitstrings, the class scores are formed by the readout \eqref{eq:head}, and the predicted label is their $\argmax$.
The reported accuracy is the fraction of the test digits whose predicted label matches the ground truth, averaged over three independently seeded federations.
The error bars are the standard error of the accuracy across the three seeds, which measures the run-to-run variability that arises from the data partition, the device assignment, and the initialization, and which is the variability an experimentalist would report over repeated deployments.

\emph{The emulator, and the path to hardware.}
The \texttt{reimei-E} emulator is the noise-inclusive software model of the \texttt{reimei} trapped-ion processor~\cite{Quantinuum}, reproducing the device native gate set, its all-to-all connectivity, and its calibrated error model, and accepting the same job specification as the physical machine.
Circuits are constructed and compiled with \texttt{pytket}~\cite{Sivarajah21} and submitted asynchronously through Quantinuum Nexus, and the submission is resumable, and a network interruption never loses completed work.
Because the submitted workflow is identical to a hardware submission, with only the backend name changing from \texttt{reimei-E} to \texttt{reimei}, the emulator results are a preview of on-device behaviour and the same experiment can be dispatched to the physical processor without any change to the code.

\bibliography{refs}

\end{document}